\let\accentvec\vec
\documentclass[10pt]{llncs}

\let\vec\accentvec
\usepackage[utf8]{inputenc}
\usepackage[T1]{fontenc}
\usepackage{lmodern}
\usepackage{subfigure}
\usepackage{multirow}
\usepackage{cellspace}
\usepackage{slashbox}
\usepackage{amssymb}
\usepackage{amsmath}
\usepackage{amsfonts}
\usepackage{latexsym}
\usepackage{listings}
\usepackage{relsize}
\usepackage[pdftex]{graphicx,color}
\usepackage{url}
 \newcommand\ForAuthors[1]
 {\par\smallskip                     
  \begin{center}
   \fbox
   {\parbox{0.9\linewidth}
    {\raggedright\sc--- #1}
   }
  \end{center}
  \par\smallskip                     %
 }
\newcommand{\comment}[1]{}

\begin{document}
\title{Polynomial Template Generation using Sum-of-Squares Programming}
\author{Assal\'e Adj\'e\inst{1}$^{,a}$ and Victor Magron\inst{2}$^{,b}$}
\institute{
Onera, the French Aerospace Lab, France.\\
Universit\'e de Toulouse, F-31400 Toulouse, France.
\\
\email{assale.adje@onera.fr}
\and
Circuits and Systems Group, Department of Electrical and Electronic Engineering,
Imperial College London, South Kensington Campus, London SW7 2AZ, UK.\\
\email{v.magron@imperial.ac.uk}
}
\renewcommand{\thefootnote}{\alph{footnote}}
\footnotetext[1]{The author is supported by the RTRA /STAE Project BRIEFCASE and the ANR ASTRID
VORACE Project.}
\footnotetext[2]{The author is supported by EPSRC (EP/I020457/1) Challenging Engineering Grant.}
\maketitle
\newcommand{\K}{\mathcal K}
\newcommand{\nn}{\mathbb N}
\newcommand{\rr}{\mathbb R}
\newcommand{\rd}{\rr^d}
\newcommand{\br}{\overline{\mathbb R}}
\newcommand{\FR}{\mathbf{F}\left(\rd,\rr\right)}
\newcommand{\F}{\mathbf{F}}
\newcommand{\fr}{\mathbf{F}\left(\pp,\br\right)}
\newcommand{\FRP}{\mathbf{F}\left(\rd,\rr_+\right)}
\def\norm#1{\mbox{$\| #1 \|$}}
\def\affect{\mathbb{A}}
\def\inter{\mathbb{I}}
\newcommand{\st}{\operatorname{s.t.}}
\newcommand{\pp}{\mathbb P}
\newcommand{\frr}{\mathbf{F}\left(\pp,\rr\right)}
\def\dual#1#2{\langle #1,#2\rangle}
\newcommand{\Id}{\operatorname{Id}}
\newcommand{\rel}[1]{#1^{\mathcal R}}
\newcommand{\sha}[1]{#1^{\sharp}}
\newcommand{\frp}{\mathbf{F}\left(\pp,\rr_{+}\right)}
\newcommand{\ve}{\operatorname{Vex_{\pp}}}
\newcommand{\vep}{\ve(\pp\mapsto\br)}
\newcommand{\vepr}{\ve(\pp\mapsto\rr)}
\newcommand{\ved}{\ve(\rd)}
\newcommand{\clo}[1]{\overline{#1}}
\newcommand{\mybrackets}[1]{\big(#1\big)}
\newcommand{\dotimes}[2]{\underset{\mbox{\rm {\scriptsize #1}}}{\overset{\mbox{\rm {\scriptsize #2}}}
{\mathlarger{\mathlarger{\otimes}}}}}
\newcommand{\Min}{\operatorname{Min}}
\newcommand{\Mk}{\rr_+^{k\times k}}
\newcommand{\well}{well-representative }
\newcommand{\prop}[1]{\mathcal{P}_{#1}}
\newcommand{\feas}[1]{\mathcal{F}\left(#1\right)}
\newcommand{\SOS}{\operatorname{SOS}}
\newcommand{\brev}{\color{red}}
\newcommand{\erev}{\color{black}}
\newcommand{\pws}{\mathcal{S}}
\def\sizefig{0.35}
\def\sizesmallfig{0.27}
\def\sizetinyfig{0.24}
\newcommand{\Sw}{\mathrm{Sw}}
\newcommand{\In}{\mathrm{In}}
\newcommand{\rea}{\mathcal{R}}
\newcommand{\inset}{\mathcal{S}}
\newcommand{\state}{\mathcal{X}}
\newcommand{\laws}{\mathcal{L}}
\newcommand{\xin}{X^{\mathrm{in}}}
\newcommand{\ind}{\mathcal{I}}
\newcommand{\cont}{\mathcal{C}}
\begin{abstract}
Template abstract domains allow to express more interesting properties than classical abstract domains.
However, template generation is a challenging problem when one uses template abstract domains for program analysis. 
In this paper, we relate template generation with the program properties that we want to prove. We focus on one-loop programs with nested conditional branches. We formally define the notion of well-representative template basis with respect to such programs and a given property. The definition relies on the fact that template abstract domains produce inductive invariants. We show that these invariants can be obtained by solving certain systems of functional inequalities. Then, such systems can be strengthened using a hierarchy of sum-of-squares (SOS) problems when we consider programs written in polynomial arithmetic. Each step of the SOS hierarchy can possibly provide a solution which in turn yields an invariant together with a certificate that the desired property holds. The interest of this approach is illustrated on nontrivial program examples in polynomial arithmetic.

\end{abstract}
\section{Introduction}
The concept of templates was introduced in a linear setting. They answered to the computational
issue of the polyhedra domain, that is, the number of faces and the number of vertices both explode when performing the code analysis. Recently, generalizations of linear templates appeared, such as 
quadratic Lyapunov functions as nonlinear templates. Nevertheless, no precise characterization of the 
templates to use have been developed for program analysis purpose. Indeed, depending on the property to show, prefixing 
a template basis without any rules can lead to unuseful information on the programs. 
For instance, suppose that we want to show that the values taken by the variables of the program
are bounded. Then, it is natural to use intervals or norm functions as templates. Unfortunately,
these functions are not sufficient to show the desired property. 
In the context of linear systems in optimal control, it is well known that Lyapunov functions 
provide useful templates to bound the variable values. This result can be extended to polynomial 
systems using polynomial Lyapunov functions. The crucial notion behind is 
that these polynomial functions allow to define sublevel sets which are invariant by the dynamics -in our case, 
the dynamics being the loop body. In static analysis, Lyapunov functions 
provide inductive invariants, which are precisely the results of computation while using template abstract domains. 
\paragraph{Related works.}Template domains were introduced by Sankaranarayanan et al.~\cite{Sriram1}, see also~\cite{Sriram2}.
The latter authors only considered a finite set of linear templates and did not provide an automatic method 
to generate templates. 
Linear template domains were generalized to nonlinear quadratic cases by Adj\'e et 
al. in~\cite{DBLP:journals/corr/abs-1111-5223,DBLP:conf/esop/AdjeGG10}, where the authors used 
in practice quadratic Lyapunov templates for affine arithmetic programs. These templates are again not automatically 
generated. Roux et al.~\cite{DBLP:conf/hybrid/RouxJGF12} provide an automatic method to compute floating-point
certified Lyapunov functions of perturbed affine loop body updates. They use Lyapunov functions with 
squares of coordinate functions as quadratic template bases in case of single loop programs written in 
affine arithmetic. The extension proposed in~\cite{AGMW13cicm,AGMW14nltemplates} relies on combining polynomial 
templates with sum-of-squares (SOS) techniques to certify nonlinear inequalities. 

Proving polynomial inequalities 
is already NP-hard and boils down to show that the infimum of a given polynomial is positive. 
However, one can obtain lower bounds of the infimum by solving a hierarchy of Moment-SOS relaxations, 
introduced by Lasserre in~\cite{Las01moments}. Recent advances in SOS optimization allowed to extensively apply 
these relaxations to various fields, including parametric polynomial optimization, optimal control, combinatorial 
optimization, {\em etc.} (see e.g.~\cite{Parrilo2003relax,LaurentSurvey} for more details). 
In the context of 
hybrid systems, certified inductive invariants can be computed by using SOS approximations of parametric polynomial 
optimization problems~\cite{Lin14hybrid}.  In~\cite{Prajna04hybrid}, the authors develop an SOS-based methodology to certify that the trajectories of hybrid systems avoid an unsafe region.
Recently, Ahmadi et al.~\cite{ahmad13switched} 
investigate necessary or sufficient conditions for SOS-convex Lyapunov functions to stabilize switched systems, 
either in the linear case or when the switched system is the convex hull of a finite number of nonlinear criteria.

In a static analysis context, polynomial invariants appear in~\cite{BagnaraR-CZ05}, where 
invariants are given by polynomial inequalities (of bounded degree) but the method relies on a reduction to 
linear inequalities (the polyhedra domain). Template polyhedra
domains allow to analyze reachability for polynomial systems: 
in~\cite{sassi2012reachability}, the authors propose a method that computes linear templates to improve the accuracy 
of reachable set approximations, whereas the procedure in~\cite{dang2012reachability} relies on Bernstein polynomials and linear programming, with linear templates being fixed in advance. Bernstein polynomials also appear in ~\cite{polynomial_template_domain} as template polynomials but there are not generated automatically. In~\cite{gulwani}, the authors use SMT-based techniques to automatically generate templates which are defined as formulas built with arbitrary logical structures and predicate conjunctions.
Other reductions to systems of polynomial {\em equalities} (by contrast with polynomial inequalities, as we consider
here) were studied in \cite{Muller,Kapur} and more recently in~\cite{cachera2014inference}.  
\paragraph{Contribution and methodology.} 
In this paper, we generate polynomial templates by combining the approach of SOS approximations extensively used in control theory with template abstract domains
originally introduced in static analysis.
We focus on analyzing programs composed of a single loop with polynomial conditional branches in the loop body and polynomial assignments. For such programs, our method consists in computing certificates which yield sufficient conditions that a given property holds. We introduce the notion of {\it well-representative} templates with respect to this property. 
Computing inductive invariant  and polynomial templates boils down to solving a system of functional inequalities. 
For computational purpose, we strengthen this system as follows:
\begin{enumerate}
\item We impose that the functions involved in each inequality of the system belong to a convex cone $\K$ included in the set of nonnegative functions. This allows in turn to define the stronger notion of {\em $\K$ well-representative} templates.
\item 
Instantiating $\K$ to the cone of SOS polynomials leads to consider a hierarchy of SOS programs, parametrized by the degrees of the polynomial templates. While solving the hierarchy, we extract polynomial template bases and feasible invariant bounds together with (SOS-based) certificates that the desired property holds.
\end{enumerate}
The potential of the method is demonstrated on several ``toy'' nonlinear programs, defined with medium-size polynomial conditionals/assignments, involving at most 4 variables and of degree up to 3. Numerical experiments illustrate the hardness of program analysis in this context, as simple nonlinear examples can already yield unexpected behaviors.
\paragraph{Organization of the paper.}The paper is organized as follows. In Section~\ref{sem-functional}, we present the programs that we want to analyze and their representation as constrained piecewise discrete-time dynamical system. Next, we recall the collecting semantics that we use and finally remind some required background about abstract semantics for generalized template domains.
Section~\ref{template-generation} contains the main contribution of the paper, namely the definition of well representative 
templates and how to generate such templates in practice using SOS programming. Section~\ref{bench} provides practical computation examples for program analysis.

\section{Static analysis context and abstract template domains}
\label{sem-functional}
In this section, we describe the programs which are considered in this paper. Next, we explain how to analyze them through their representation as discrete-time dynamical systems. Then, we give details about the special properties which can be inferred
on such programs. Finally,  we recall mandatory results for abstract template domains that are used in the sequel of the paper.
\subsection{Program syntax and constrained piecewise discrete-time dynamical system representations}
In this paper, we are interested in analyzing computer science programs. We focus on programs composed of a single loop with a possibly complicated switch-case type loop body. This  
loop is supposed to be written as a nested sequence of \emph{if} statements.
Moreover we suppose that the analyzed programs are written in Static Single Assignment (SSA) form, that is each variable is initialized at most once. We denote by $(x_1,\ldots,x_d)$ the vector of the program variables.
Finally, we consider assignments of variables using only {\em parallel assignments} $(x_1,\ldots,x_d)=T(x_1,\ldots,x_d)$. Tests are either weak inequalities $r(x_1,\ldots,x_d) \leq 0$ or strict inequalities $r(x_1,\ldots,x_d) < 0$. We assume that assignments are functions from $\rd$ to $\rd$ and test functions are functions from $\rd$ to $\rr$. In the program syntax, the notation $\ll$ will be either $\verb+<=+$ or $\verb+<+$. The form of the analyzed program is described in Figure~\ref{programstyle}. 
\begin{figure}[!ht]
\begin{center}
\begin{tabular}{|c|}
\hline
\begin{lstlisting}[mathescape=true]
x $\in$ $\xin$;
while ($r_1^0$(x)$\ll$0 and ... and $r_{n_0}^0$(x)$\ll$0){
  if($r_1^1$(x)$\ll$0){ 
     $\vdots$
     if($r_{n_1}^1$(x)$\ll$0){
        x = $T^1$(x);
     }   
     else{
        $\vdots$
        if($r_{n_i}^i$(x)$\ll$0){
           x = $T^i$(x);   
        }   
     }
  else{
      $\vdots$
  } 
}
\end{lstlisting}\\
\hline
\end{tabular}
\end{center}
\caption{One-loop programs with nested conditional branches}
\label{programstyle}
\end{figure}

As depicted in Figure~\ref{programstyle}, an update $T^i : \rd \to \rd$ of the $i$-th condition branch is executed if and only if the conjunction of tests $r_j^i(x)\ll 0$ holds. 
The variable $x$ is updated by $T^i(x)$ if the current value of $x$ belongs to $X^i:= \{x \in \rd \, | \,  \forall j = 1,\dots,n_i, \ r_j^i(x)\ll 0 \}$.           
Consequently, we interpret programs as \emph{constrained piecewise discrete-time dynamical systems} (CPDS for short). 
The term \emph{piecewise} means that there exists a partition $\{X^i,i\in \ind\}$ of $\rd$ such that for all $i\in \ind$, 
the dynamics of the system is represented by the following relation, for $k\in\nn$:
\begin{equation}
\label{pws}
\text{if } x_k\in X^i\cap X^0,\ x_{k+1}=T^i(x_k) \,.
\end{equation}
We assume that the initial condition $x_0$ belongs to some compact set $\xin$. For the program, $\xin$ is the set where the variables are supposed to be initialized in. Since the test entry for the loop condition can be nontrivial, we add the term \emph{constrained} and $X^0$ denotes the set representing the conjunctions of tests for the loop condition. The iterates of the CPDS are constrained to live in $X^0$: if for some step $k\in\nn$, $x_k\notin X^0$ then the  CPDS is stopped at this iterate with the terminal value $x_k$.
We define a partition as a family of nonempty sets such that:
\begin{equation}
\label{partition}
\bigcup_{i\in\ind} X^i=\rd,\ \forall\, i,j\in\ind,\ i\neq j, X^i\cap X^j\neq \emptyset \,.
\end{equation}
From Equation~\eqref{partition}, for all $k\in\nn^*$ there exists a unique $i\in\ind$ such that $x_k\in X^i$. A set $X^i$ can contain both strict and weak inequalities and characterizes the set of the $n_i$ conjunctions of tests functions $r_j^i$. Let $r^i=(r_1^i,\ldots,r_{n_i}^i)$ stands for the vector of tests functions associated to the set $X^i$. Moreover, for $X^i$, we denote by $r^{i,s}$ (resp. $r^{i,w}$) the part of $r^i$ corresponding to strict (resp. weak) inequalities. Finally, we obtain the representation of the set $X^i$ given by Equation~\eqref{semialgebraic}:
\begin{equation}
\label{semialgebraic}
X^i=\left\{x\in\rd \left| r^{i,s}(x) < 0,\ r^{i,w}(x)\leq 0\right\}\right. \,.
\end{equation}
We insist on the notation: $y< z$ (resp. $y_l<z_l$) means that for all coordinates $l$, $y_l<z_l$ (resp. $y_l\leq z_l$).

We suppose that the sets $\xin$ and $X^0$ also admits the representation given by Equation~\eqref{semialgebraic} and we denote by $r^0$ the vector of tests functions $(r_1^0,\ldots,r_{n_0}^0)$ and by $r^{\mathrm{in}}$ the vector of tests functions $(r_1^{\mathrm{in}},\ldots,r_{n_{\mathrm{in}}}^{\mathrm{in}})$. We also decompose $r^0$ and $r^{\mathrm{in}}$ as strict and weak inequality parts denoted respectively by $r^{0,s}$, $r^{0,w}$, $r^{\mathrm{in},s}$ and $r^{\mathrm{in},w}$.
To sum up, we give a formal definition of CPDS.
\begin{definition}[CPDS]
\label{pwsdef}
A constrained piecewise discrete-time dynamical system (CPDS) is the quadruple $(\xin,X^0,\state,\laws)$ with:
\begin{itemize}
\item $\xin\subseteq \rd$ is the compact of the possible initial conditions;
\item $X^0\subseteq \rd$ is the set of the constraints which must be respected by the state variable;
\item $\state:=\{X^i, i\in \ind\}$ is a partition as defined in Equation~\eqref{partition};
\item $\laws:=\{T^i, i\in\ind\}$ is the family of the functions from $\rd$ to $\rd$, w.r.t. the partition $\state$ satisfying Equation~\eqref{pws}.
\end{itemize}   
\end{definition}
From now on, we associate a CPDS representation to each program of the form described at Figure~\ref{programstyle}. 
Since a program admits several CPDS representations, we choose one of them, but this arbitrary choice does not change the results 
provided in this paper. 
In the sequel, we will often refer to the running example described in Example~\ref{running}.
\begin{example}[Running example]
\label{running}
The program below involves four variables and contains 
an infinite loop with a conditional branch in the loop body.
The update of each branch is polynomial. The parameters $c_{i j}$ (resp.  $d_{i j}$) are given parameters.
During the analysis, we only keep the variables $x_1$ and $x_2$ since 
$oldx_1$ and $oldx_2$ are just memories.
\begin{center}
\begin{tabular}{c}
\begin{lstlisting}[mathescape=true]
$x_1, x_2\in [a_1, a_2] \times [b_1, b_2]$;
$oldx_1$ = $x_1$;
$oldx_2$ = $x_2$;
while (-1 <= 0){
  if ($oldx_1$^2 + $oldx_2$^2 <= 1){ 
      $oldx_1$ = $x_1$;
      $oldx_2$ = $x_2$;
      $x_1$ = $c_{11}$ * $oldx_1$^2 + $c_{11}$ * $oldx_2$^3;
      $x_2$ = $c_{21}$ * $oldx_1$^3 + $c_{22}$ * $oldx_2$^2;
  }
  else{
      $oldx_1$ = $x_1$;
      $oldx_2$ = $x_2$;
      $x_1$ = $d_{11}$ * $oldx_1$^3 + $d_{12}$ * $oldx_2$^2;
      $x_2$ = $d_{21}$ * $oldx_1$^2 + $d_{22}$ * $oldx_2$^2;
  } 
}
\end{lstlisting}
\end{tabular}
\end{center}
Its constrained piecewise discrete-time dynamical system representation corresponds to the quadruple $(\xin,X^0,\{X^1,X^2\},\{T^1,T^2\})$, where the set of initial conditions is:
\[ 
\xin= [a_1, a_2] \times [b_1, b_2] \,,
\]
the set $X^0$ in which the variable $x=(x_1,x_2)$ lies is:
\[
X^0=\rd \,,
\]
the partition verifying Equation~\eqref{partition} is:
\[
X^1=\{x\in\rr^2\mid x_1^2+x_2^2\leq 1\},\quad X^2=\{x\in\rr^2\mid -x_1^2-x_2^2< -1\} \,,
\]
and the functions relative to the partition $\{X^1,X^2\}$ are:
\[
T^1(x)=\left(\begin{array}{c}
c_{1 1} x_1^2 + c_{1 2} x_2^3\\
c_{2 1} x_1^3 + c_{2 2} x_2^2
\end{array}\right)
\text{ and }  
T^2(x)=\left(\begin{array}{c}
d_{1 1} x_1^3 + d_{1 2} x_2^2\\
d_{2 1} x_1^2 + d_{2 2} x_2^2
\end{array}\right)\enspace.
\]
\end{example}
\subsection{Program invariants}
The main goal of the paper is to decide automatically if a given property holds for the analyzed program. We are interested in numerical properties and more precisely in properties on the values taken by the $d$-uplet of the variables of the program. Hence, 
in our point-of-view, a property is just the membership of some set $P\subset \rd$. In particular, we study properties which are valid after an arbitrary number of loop iterates. Such properties are called \emph{loop invariants} of the program. Formally, we use the CPDS representation of a given program and we say that $P$ is a loop invariant of this program if: 
\[
\forall\, k\in\nn,\ x_k\in P \,,
\] 
where $x_k$ is defined at Equation~\eqref{pws} as the state variable at step $k\in\nn$ of the CPDS representation of the program. 

Now, let us consider a program of the form described in Figure~\ref{programstyle} and let us denote by $\pws$ the CPDS representation of this program. The set $\rea(\pws)$ of \emph{reachable values} is the set of all possible values taken by the state variable along the running of $\pws$. We define $\rea(\pws)$ as follows:
\begin{equation}
\label{reachdef}
\rea(\pws)=\{ y\in\rd\mid \exists\ k\in \nn, \exists\ i\in \ind,\ x_k\in X^i\cap X^0,\ y=T^i(x_k)\}\cup \xin \,.
\end{equation}
To prove that a set $P$ is a loop invariant of the program is equivalent to prove that $\rea(\pws)\subseteq P$. 
We can rewrite $\rea(\pws)$ by introducing auxiliary variables $\rea^i$, $i\in\ind$:
\begin{equation}
\label{auxsemantics}
\rea(\pws)=\bigcup_{i\in\ind} \rea^i\cup \xin,\ \rea^i=T^i\left(\rea(\pws)\cap X^i\cap X^0\right) \,.
\end{equation}
Let us denote by $\wp(\rd)$ the set of subsets of $\rd$ and introduce the map $F: \left(\wp(\rr^d)\right)^{|\ind|+1} \rightarrow \left(\wp(\rr^d)\right)^{|\ind|+1}$ defined by:
\begin{equation}
\label{transferfunctional}
F_{i}(C_1,\ldots,C_{|\ind|+1})
=\left\{
\begin{array}{lr}
T^i\left(C_{|\ind|+1}\cap X^i\cap X^0\right)& \text{ if } j\neq |\ind|+1 \,,\\
\bigcup_{k\in\ind} C_k\cup \xin & \text{ otherwise} \,.
\end{array}
\right.
\end{equation}
We equip $\wp(\rr^d)$ with the partial order of inclusion and $\left(\wp(\rr^d)\right)^{|\ind|+1}$ by the standard component-wise partial order. The infimum is understood in this sense i.e. as the greatest lower bound with respect to this order. The smallest fixed point problem is:
\begin{equation*}
\inf \left\{\mathbf{C}=(C_1,\ldots,C_{|\ind|+1})\in\left(\wp(\rr^d)\right)^{|\ind|+1}\\
\mid \forall\, i= 1,\ldots,|\ind|+1, C_i=F_i(\mathbf{C})\right \}\,.
\end{equation*}
It is well-known from Tarski's theorem that the solution of this problem exists, is unique and in this case, it corresponds to $(\rea^1,\rea^2,\ldots,\rea(\pws))$ where $\rea^1,\rea^2...\rea^{|\ind|}$ are defined in Equation~\eqref{auxsemantics}. Tarski's theorem also states that $(\rea^1,\rea^2,\ldots,\rea(\pws))$ is the smallest solution of the following Problem:
\begin{equation*}
\inf \left\{\mathbf{C}=(C_1,\ldots,C_{|\ind|+1})\in\left(\wp(\rr^d)\right)^{|\ind|+1}\\
\mid \forall\, i= 1,\ldots,|\ind|+1, F_i(\mathbf{C})\subseteq C_i\right\} \,.
\end{equation*}

We warn the reader that the construction of $F$ is completely determined by the data of the CPDS $\pws$. But for the sake of conciseness, we do not make it explicit on the notations. Note also that the map $F$ corresponds to a standard transfer function (or collecting semantics functional) applied to the CPDS representation of a program.
\begin{example}[Transfer function of the running example]
Since $X^0=\rd$, the transfer function $F$ associated to the CPDS of Example~\ref{running} is given by: 
\[
\begin{array}{c}
F_1(C_1,C_2,C_3)=T^1(C_3\cap X^1) \,,\\ 
F_2(C_1,C_2,C_3)=T^2(C_3\cap X^2) \,,\\ 
F_3(C_1,C_2,C_3)=C_1\cup C_2\cup \xin \,.
\end{array}
\]
\end{example}
To prove that a subset $P$ is a loop invariant, it suffices to show that $\mathbf{P}=(T^1(P\cap X^1\cap X^0),\ldots,P)$ satisfies $F_{|\ind|+1}(\mathbf{P})\subseteq P$. Nevertheless, $F$ is still not computable and we use abstract interpretation~\cite{CC77} to provide safe over-approximations of $F$. Next, we use generalized abstract template domains as abstract domains and we construct a safe over-approximation of $F$ using a Galois connection. 
In this paper, we consider invariants defined from properties which are encoded with sublevel sets of given functions. A loop invariant is supposed to be the union of sublevel sets of a given function from $\rd$ to $\rr$. 
\begin{definition}[Sublevel property]
\label{funproperty}
Given a function $\kappa$ from $\rd$ to $\rr$, we define the sublevel property $\prop{\kappa}$ as follows:
\[
\prop{\kappa}:=\bigcup_{\alpha\in\rr} \{x\in\rd\mid \kappa(x)\leq \alpha\} \,.
\]
\end{definition}
\begin{example}[Sublevel property examples]
\begin{enumerate}
\item Let $\kappa$ be a norm on $\rd$, then $\prop{\kappa}$ is the property 
``the values taken by the variables are bounded''.
\item Let $\kappa : x \mapsto x_i$,
then $\prop{\kappa}$ is the property ``the values taken by the variable 
$x_i$ are bounded from above''.
\item We can ensure that the set of possible values taken by the program variables avoids an unsafe region
with a fixed level sublevel property. For example, if the property to show consists in proving that the square norm of the variable is still greater than 1, we can set $\kappa(x)=1- \|x \|_2^2$ and restrict the sublevel sets to those for which  $\alpha\leq 0$.
\end{enumerate}
\end{example}
A sublevel property is called \emph{sublevel invariant} when this property is a loop invariant.
We describe how to construct template bases, so that we can prove that a sublevel property is a sublevel invariant.
\if{
Given a property, we pick a relevant template basis w.r.t. this property.
As the function  $\kappa$ related to a given property does not take into account the dynamics of the program, it does not seem to be relevant w.r.t of the template basis. 
}\fi
\subsection{Abstract template domains}
\label{template-domain}
The concept of generalized templates was introduced in~\cite{DBLP:conf/esop/AdjeGG10,DBLP:journals/corr/abs-1111-5223}. 
Let $\FR$ stands for the set of functions from $\rd$ to $\rr$.
\begin{definition}[Generalized templates]
\label{templatedef}
A generalized template $p$ is a function from $\rd$ to $\rr$ over the vector of variables $(x_1,\ldots,x_d)$.
\end{definition}
Templates can be viewed as implicit functional relations on variables to prove certain properties on the analyzed program. 
We denote by $\pp$ the set of templates. 
First, we suppose that $\pp$ is given by some oracle and say that $\pp$ forms a template basis. 
Here, we recall the required background about generalized templates (see~\cite{DBLP:conf/esop/AdjeGG10,DBLP:journals/corr/abs-1111-5223} for more details). 
\subsubsection{Basic notions}
We replace the classical concrete semantics by meaning of sublevel sets i.e.
we have a functional representation of numerical invariants through the functions of $\pp$. An invariant
is determined as the intersection of sublevel sets. The problem is thus reduced to find 
optimal level sets on each template $p$. Let $\fr$ stands for the set of functions from $\pp$ to 
$\br=\rr\cup\{-\infty\}\cup\{+\infty\}$. 
\begin{definition}[$\pp$-sublevel sets]
\label{concretisation}
For $w\in\fr$, we associate the $\pp$-sublevel set $w^{\star}\subseteq\rd$ given by:
\[
\displaystyle{w^{\star}=\{x\in \rd\mid p(x)\leq w(p),\ \forall p\in\pp \}=\bigcap_{p\in\pp}\{x\in\rd\mid p(x)\leq w(p)\}} \, .
\]
\end{definition}
\if{
When $\pp$ is a set of convex functions, the $\pp$-sublevel sets corresponds to the intersection of classical 
sublevel sets from convex analysis. In our case, $\pp$ can contain non-convex functions so $\pp$-sublevel sets 
are not necessarily convex in the usual sense. 
}\fi
In convex analysis, a closed convex set can be represented 
by its support function i.e. the supremum of linear forms on the set (e.g. ~\cite[\S~13]{Roc}). 
Here, we use the generalization by Moreau~\cite{Moreau} (see also~\cite{Rubinov,Singer}) which consists in replacing the linear forms by the functions $p\in\pp$. 
\begin{definition}[$\pp$-support functions]
\label{ppsupport}
To $X\subseteq\rd$, we associate the abstract support function denoted by
$X^{\dag}:\pp\mapsto \br$ and defined by:
\[
X^{\dag}(p)=\sup_{x\in X} p(x) \, .
\]
\end{definition}  
Let $C$ and $D$ be two ordered sets equipped respectively by the order $\leq_C$ and $\leq_D$. Let $\psi$ be a map from $C$ to $D$ and $\varphi$ be a map from $D$ to $C$. We say that the pair $(\psi,\varphi)$ defines a Galois connection between $C$ and $D$ if and only if $\psi$ and $\varphi$ are monotonic and the equivalence $\psi(c)\leq_D d\iff \varphi(d)\leq_C c$ holds for all $c\in C$ and all $d\in D$.   

We equip $\fr$ with the partial order of real-valued
functions i.e. $w\leq v\iff w(p)\leq v(p)\ \forall p\in \pp$. The  set $\wp(\rd)$ is equipped with 
the inclusion order.
\begin{proposition}
\label{Galois}
The pair of maps $w\mapsto w^{\star}$ and $X\mapsto X^{\dag}$ defines a 
Galois connection between $\fr$ and the set of subsets of $\rd$.
\end{proposition}
In the terminology of abstract interpretation, $(\cdot)^\dag$ is the abstraction function,
and $(\cdot)^\star$ is the concretisation function. The Galois connection result provides the 
correctness of the semantics. We also remind the following property:

\begin{equation}
\label{galoisprop}
(((w^\star)^\dag)^\star=w^\star \,, \qquad  ((X^\dag)^\star)^\dag=X^\dag \, .
\end{equation}
\subsubsection{The lattices of $\pp$-convex sets and $\pp$-convex functions}
Now, we are interested in closed elements (in term of Galois connection), called 
$\pp$-convex elements.
\begin{definition}[$\pp$-convexity]
\label{abstractconvexity}
Let $w\in\fr$, we say that $w$ is a $\pp$-convex 
function if $w=(w^\star)^{\dag}$.
A set $X\subseteq\rd$ is a $\pp$-convex set if $X=(X^{\dag})^{\star}$.
We respectively denote by $\vep$ and $\ved$ the set of $\pp$-convex functions of $\fr$ and
the set of $\pp$-convex sets of $\rd$.
\end{definition}
The family of functions $\vep$ is ordered by the partial order of real-valued
functions. The family of sets $\ved$ is ordered by the inclusion order.
Galois connection allows to construct lattice operations on $\pp$-convex elements. 
\begin{definition}[The meet and join]
Let $v$ and $w$ be in $\fr$. We denote by $\inf(v,w)$
and $\sup(v,w)$ the functions defined respectively by, 
$p\mapsto\inf(v(p),w(p))$ and $p\mapsto\sup(v(p),w(p))$.
We equip $\vep$ with the join operator $v\vee w=\sup(v,w)$ and the meet 
operator $v\wedge w =(\inf(v,w)^{\star})^{\dag}$.
Similarly, we equip $\ved$ with the join operator 
$X\sqcup Y=((X\cup Y)^{\dag})^{\star}$ and the meet operator $X\sqcap Y =X\cap Y$.
\end{definition}
The next theorem follows readily from the fact that the pair of $v\mapsto v^\star$ and $C\mapsto C^\dag$ defines a Galois
connection (see e.g.~\cite[\S~7.27]{priestley}).
\begin{theorem}
\label{lattice}
The complete lattices $(\vep,\wedge,\vee)$ and $(\ved,\sqcap,\sqcup)$ are isomorphic. 
\end{theorem}
\subsubsection{Abstract semantics}
Since the pair of maps $w\mapsto w^{\star}$ and $X\mapsto X^{\dag}$ is a Galois connection 
(Proposition~\ref{Galois}), we can construct abstract semantics functional from this pair and the map $F$ defined at Equation~\eqref{transferfunctional}. We obtain a map $\sha{F}$ from $\vep^{|\ind|+1}$ to itself defined for $w\in\vep^{|\ind|+1}$ and $p\in\pp$ by:
\begin{equation*}
\label{abstraction}
\mybrackets{\sha{F_{i}}(w)}(p)=
\left\{
\begin{array}{c}
\displaystyle{\sup_{y\in T^i\left( w_{|\ind|+1}^{\star}\cap X^i\cap X^0\right)} p(y)}
=\displaystyle{\sup_{\substack{x\in w_{|\ind|+1}^{\star}\\ r_s^i(x)< 0,\ r_w^i(x)\leq 0\\ r_s^0(x)<0,\ r_w^0(x)\leq 0}} p(T^i(x))}\\
\displaystyle{\sup_{y\in \bigcup_{j\in\ind} w_{j}^{\star}\cup \xin} p(y)}
=\left(\bigcup_{j\in\ind} w_{j}^{\star}\cup \xin\right)^{\dag}(p)
\end{array}
\right.
\end{equation*}
Since $F$ is conditioned by the data of the CPDS $\pws$, it is also the case for $\sha{F}$. 
As a corollary of Theorem~\ref{lattice}, the best abstraction of $\rea(\pws)$ in the lattice $\vep$ 
is the smallest fixed point of Equation~\eqref{loopinvariantabs}. 
\begin{equation}
\label{loopinvariantabs}
\inf \left\{
\begin{array}{l}
\mathbf{w}=(w_1,\ldots,w_{|\ind|+1})\in\vep^{|\ind|+1}\\
\st\ \forall\, i=1,\ldots,|\ind|+1,\ \sha{F_i}(\mathbf{w})\leq w_i \,.
\end{array}
\right\}
\end{equation}
The infimum is understood in the sense of the order of the component-wise order of the complete lattice $\vep^{|\ind|+1}$. Using Tarski's theorem, the solution of Equation~\eqref{loopinvariantabs} exists and is unique and is usually called the abstract semantics. This latter solution is optimal but any feasible solution could provide an answer to decide whether a sublevel property is an invariant of the program. 
\begin{definition}[Feasible invariant bound]
\label{feasibleinvariantbound}
The function $w\in\vep$ is a feasible invariant bound w.r.t. to the CPDS $\pws=(\xin,X^0,\{X^i,i\in\ind\}, \{T^i,i\in\ind\})$ 
iff it exists $(w_1,\ldots,w_{|\ind|})\in\vep^{|\ind|}$ such that:
\begin{equation}
\label{feasibleequation}
w\geq \sup\{{\xin}^\dag,\sup_{i\in\ind} w_i\}\wedge\left(\ \forall\, i\in\ind,\ w_i\geq \left(T^i\left(w^\star\cap X^i\cap X^0\right)\right)^\dag\right)
\end{equation}
In the sequel, we denote by $\feas{\pws}$ the set of feasible invariant bounds.
\end{definition}
From the definition of feasible invariant bound, we state the following proposition.
\begin{proposition}
\label{basicresults}
Let us consider a CPDS $\pws=(\xin,X^0,\{X^i,i\in\ind\}, \{T^i,i\in\ind\})$.
The following statements are true:
\begin{enumerate}
\item Let $(w_1,\ldots,w_{|\ind|+1})$ be a solution of Problem~\eqref{loopinvariantabs}, then $w_{|\ind|+1}$ is the smallest feasible invariant bound w.r.t. $\pws$;
\item For all $w\in\feas{\pws}$, $\rea(\pws)\subseteq w^\star$. 
\end{enumerate}
\end{proposition}
For a given program represented by the CPDS $\pws$, we recall that an invariant $P\subset \rd$ is to said be an \emph{inductive invariant} of this program if for all $k\in\nn$, the implication $x_k\in P\implies x_{k+1}\in P$ holds for the state variable $x_k$. 
Next, for a given function $w\in\vep$, we give a simple condition in term of inductive invariants (up to test functions) for $w$ to be a feasible invariant bound. 
\begin{proposition}[Loop head invariants in template domains]
\label{inductive}
Let us consider the CPDS $\pws=(\xin,X^0,\{X^i,i\in\ind\}, \{T^i,i\in\ind\})$ and $w\in\vep$. Suppose that:
\if{
\begin{enumerate}
\item $ $ ;
\item $ $ ;
\item $ $.
\end{enumerate}
}\fi
\begin{equation}
\label{implication}
\xin\subseteq w^\star \wedge \left(\forall\, i\in\ind\ \left(x\in w^{\star}\wedge x\in X^i\wedge x\in X^0\implies T^i(x)\in w^{\star}\right)\right) \,.
\end{equation}
Then $w\in\feas{\pws}$.
\end{proposition}
\begin{proof}
From the definition of the $(\cdot)^\dag$ operator and Proposition~\ref{Galois}, Conjunction~\eqref{feasibleequation} holds with $w_i=w$ for all $i\in\ind$.
\if{
By Proposition~\ref{Galois}, the first inequality implies that $C^{\dag}\leq w$. Moreover using the 
implications involving $T^i$ and $T^e$, we get $\forall\, p\in\pp$,
$w(p)\geq \displaystyle{\sup_{x\in w^{\star}\wedge r(x)\leq 0\wedge s(x)\leq 0} p(T^i(x))}$
and $w(p)\geq \displaystyle{\sup_{x\in w^{\star}\wedge r(x)\leq 0\wedge s(x)\geq 0} p(T^e(x))}$
which implies $w(p)\geq \displaystyle{\sup_{x\in w^{\star}\wedge r(x)\leq 0\wedge s(x)> 0} p(T^e(x))}$.
This is nothing but $w\geq \left(T^i\left(w^\star \cap r^{-1}((-\infty,0])\cap s^{-1}((-\infty,0])\right)\right)^\dag$
and $w\geq \left(T^e\left(w^\star \cap r^{-1}((-\infty,0])\cap s^{-1}((-\infty,0])\right)\right)^\dag$.
From definition of suprema, we conclude that $w\in\vep$ is a feasible invariant bound w.r.t. 
$Pr(C,r,s,T^i,T^e)$ by simply taking $v_2=w,v_3=w,v_4=w$.
}\fi
\qed
\end{proof}
We recalled that abstract template domains produce invariants, i.e. $\pp$-sublevel sets of feasible invariant bounds. It is not surprising since abstract template domains are abstract domains. The main issue is that $\pp$ is supposed to be given. The question is which templates basis $\pp$ can produce a nontrivial (strictly smaller that $\rd$) feasible invariant bound? This question can be refined when we want to show that some sublevel property is an invariant: which templates basis can ensure that the sublevel property is an invariant of the program? We propose an answer by considering Equation~\eqref{implication} as a system of equations, where unknowns are the template basis $\pp$
and $w\in\vep$. Given a sublevel $\prop{\kappa}$, we also impose that $w$ and $\pp$ satisfy $w^\star\subseteq \prop{\kappa}$. This latter constraint leads to the computation of a level $\alpha$ for which $\{x\in\rd \mid\kappa(x)\leq \alpha\}$ is an invariant of the program.  
\section{Proving program properties using sum-of-squares}
\label{template-generation}
Here, we describe how to certify that a sublevel property is a loop invariant using sum-of-squares (SOS) approximations. In Section~\ref{general}, we provide a formal definition of the set of template bases that we shall use to the latter certification.
Then we describe how to construct template bases so that we can prove sublevel properties (Section~\ref{simplemethods}). 
In the end, we explain how to compute such bases in practice, by solving a hierarchy of SOS programs (Section~\ref{sos}).
\subsection{The general setting}
\label{general}
\begin{definition}[Well-representative template basis w.r.t. a CPDS and a sublevel property]
\label{boundwellrepre}
Let $\prop{\kappa}$ be a sublevel property and $\pws=(\xin,X^0,\{X^i,i\in\ind\},\{T^i,i\in\ind\})$ be a CPDS. The template basis $\pp$ is  \well w.r.t. $\pws$ and $\prop{\kappa}$ iff there exists $w\in\feas{\pws}$ such that $w^\star\subseteq \prop{\kappa}$.
\end{definition}
In the sequel, we fix a CPDS $\pws=(\xin,X^0,\{X^i,i\in\ind\},\{T^i,i\in\ind\})$ and a sublevel property $\prop{\kappa}$.

Well-representative template bases explicit the sets of implicit functional relations on the program variables, 
needed to prove that a sublevel property is an invariant. Next, we define a cone structure to strengthen the notion
of well-representative bases. 
\if{
Let $\K$ be a convex cone, included in the set $\FRP$ of nonnegative functions on $\rd$. 
 bases easier to handle: the $\K$ well-representative template bases. To define this new notion,
we use a convex cone $\K$ of the set of functions on $\rd$ which take nonnegative values. We give a formal
definition of this kind of convex cones at Definition~\ref{convexcones}.
From now, we denote by $\FRP$ the set of functions on $\rd$ which have nonnegative values.
The operator $+$ and the scalar multiplication are understood in the functional sense. 
}\fi
\begin{definition}[Convex cones containing the scalars in $\FRP$]
\label{convexcones}
A non-empty subset $\K$ of $\FRP$ is a convex cone containing the scalars iff:
 \begin{enumerate}
\item for all $f\in\K$, for all $t\geq 0$, $tf\in\K$;
\item for all $f,g\in \K$, $f+g\in \K$;
\item for all $c\in \rr_+$, $x\mapsto c\in\K$;
\end{enumerate}
\end{definition}
In the sequel, we write $c\in\K$ instead of $x\mapsto c\in\K$, for each $c\in\rr_+$.     
For a convex cone containing the scalars $\K$, $\K^k$ stands for the set of 
vectors of $k$ elements of $K$ and $\K^{n\times k}$ stands for the set of tableaux of $n\times k$ elements of $\K$. 
For $\lambda\in \K^{n\times k}$, we denote the ``row m'' of $\lambda$ by $\lambda_{m,\cdot}$ 
and the ``column j'' of $\lambda$ by $\lambda_{\cdot,j}$. Thus $\lambda_{m,j}$ refers 
to the $m,j$ element of the tableau $\lambda$.
 
We derive a stronger notion of well-representative 
template bases, namely $\K$ well-representative template bases  
This notion is more restrictive, as
a $\K$ well-representative template 
basis deals with a system of inequalities instead of conjunctions of implications. 
\begin{definition}[$\K$ well-representative template basis]
\label{kwell}
A finite template basis $\pp=\{p_1,\ldots,p_k\}$ is a $\K$ well-representative template basis w.r.t. 
$\pws$ and $\prop{\kappa}$ iff there exist $w\in\rr^k$, $\alpha\in\rr$, $\nu\in\K^k$ and for all $i\in\ind$, there exist $\lambda^i\in\K^{k\times k}$, $\mu^i\in \K^{k\times n_i}$, $\gamma^i\in
\K^{k\times n_0}$ such that: 
\begin{enumerate}
\item Initial condition satisfiability: $\forall\, l=1,\ldots, k$,
\[
w_l\geq \sup_{y\in \xin} p_l(y) \, .
\]
\item ``Local'' branch satisfiability: $\forall\, l=1,\ldots, k$, $\forall\, i\in\ind$:
\[
w_l-\sum_{j=1}^k \lambda_{l,j}^i(x) (w_j-p_j(x))-p_l(T^i(x))+\sum_{j=1}^{n_i}\mu_{l,j}^i(x)r_j^i(x)+\sum_{j=1}^{n_0}\gamma_{l,j}^i(x)r_j^0(x)\in \K \, .
\]
\item Property satisfiability: 
\[\alpha -\kappa(x) -\sum_{t=1}^k \nu_{t}(x) (w_t-p_t(x))\in \K \enspace.\]
\end{enumerate}
\end{definition}
For the sake of presentation, let us define for all $l=1,\ldots, k$, for all $i\in\ind$:
\begin{equation}
\label{auxiliaryineq}
\left\{
\begin{array}{l}
\displaystyle{S_l^i:x\mapsto}\\
\displaystyle{w_l-\sum_{j=1}^k \lambda_{l,j}^i(x)
(w_j-p_j(x))-p_l(T^i(x))+\sum_{j=1}^{n_i}\mu_{l,j}^i(x)r_j^i(x)+\sum_{j=1}^{n_0}\gamma_{l,j}^i(x)r_j^0(x)}\enspace, \\
\displaystyle{S^\kappa:x\mapsto \alpha -\kappa(x) -\sum_{t=1}^k \nu_{t}(x) (w_t-p_t(x))}  \enspace. 
\end{array}
\right.
\end{equation}
\begin{example}[$\K$ well-representative template basis]
\label{strongex}
Consider Example~\ref{running}. We are interested in proving the boundedness of the values taken 
by the variables of the program. For $x=(x_1,x_2)$, let consider $\kappa(x)=\norm{x}_2^2=x_1^2+x_2^2$. Recall that
$\xin= [a_1, a_2] \times [b_1, b_2],\ 
X^0=\rd,\ X^1=\{x\in\rr^2\mid x_1^2+x_2^2\leq 1\},\quad X^2=\{x\in\rr^2\mid -x_1^2-x_2^2< -1\} $, $T^1(x_1,x_2)=(c_{11}x_1^2+c_{12}x_2^3,c_{21}x_1^3+c_{22}x_2^2)$
and $T^2(x_1,x_2)=(d_{11}x_1^3+d_{12}x_2^2,d_{21}x_1^2+d_{22}x_2^2)$. Let $\K=\FRP$ and $\{p\}$ be a singleton template basis. Then $\{p\}$ is $\K$ well-representative w.r.t. the CPDS $(\xin,X^0,\{X^1,X^2\},\{T^1,T^2\})$ and $\prop{\kappa}$
iff there exists $w\in\rr$, $\alpha\in\rr_+$, $\nu\in\FRP$, $\lambda^1,\lambda^2\in\FRP$ and $\gamma^1,\gamma^2\in\FRP$ such that:
\[
\left\{
\begin{array}{l}
w\geq \sup_{y\in [a_1, a_2] \times [b_1, b_2]} p(y) \,,\\
\forall x\in\rr^2,\ w-\lambda^1(x) (w-p(x))-p(T^1(x))
+\gamma^1(x)(\norm{x}_2^2-1)\geq 0 \, ,\\
\forall x\in\rr^2,\ w-\lambda^2(x) (w-p(x))-p(T^2(x))
+\gamma^2(x)(1-\norm{x}_2^2)\geq 0 \, ,\\
\forall x\in\rr^2,\ \alpha -\norm{x}_2^2 -\nu(x) (w-p(x))\geq 0 \, .
\end{array}
\right.
\]
\end{example}
Note that generating inductive invariants is well known to yield undesirable nonlinear optimization problems (e.g. bilinearity, as in~\cite{Sriram3}). Here nonlinearity is avoided by fixing the parameters $\{\lambda^i,i\in\ind\}\subseteq\K^{k\times k}$ and $\nu\in \K^k$ to 1, so that 
the two last inequalities of Definition~\ref{kwell} become linear in the variables $p_1,\ldots,p_k$, $w_1,\ldots,w_k$,
$\alpha$ and the parameters $\{\mu^i,i\in\ind\},\{\gamma^i,i\in\ind\}\in\K^k$.  
\if{
But the next proposition Proposition~\ref{linear} tends to prove that to fix certain parameters 
makes the problem easier and even linear in the template bases. 
\begin{proposition}[Linearisation of $\K$ well-representative conditions]
\label{linear}
\begin{enumerate}
\item For all $\lambda^i,\lambda^e\in\K^{k\times k}$ and for all $l=1,\ldots,k$: 
\[
\begin{array}{l}
\displaystyle{(\pp,\mu^i,\gamma^i,w)\mapsto w_l-\sum_{j=1}^k \lambda_{l,j}^i(x) (w_j-p_j(x))-p_l(T(x))
+\mu^i(x)r(x)+\gamma^i(x)s(x)}\\
\text{and}\\
\displaystyle{(\pp,\mu^e,\gamma^e,w)\mapsto w_l-\sum_{j=1}^k \lambda_{l,j}^e(x) (w_j-p_j(x))-p_l(T(x))
+\mu^e(x)r(x)+\gamma^e(x)s(x)}
\end{array}
\]
is linear (w.r.t. $\FR^{k+1+1}\times \rr^k$).  
\item For all $\nu\in\K^{n\times k}$ and for all $j=1,\ldots,k$:
\[
(\pp,\alpha,w)\mapsto \alpha -\kappa(x) -\sum_{l=1}^k \nu_j(x) (w_j-p_j(x))
\]
is linear (w.r.t. $\FR^{k}\times \rr^{1+k}$) .
\end{enumerate}
\end{proposition}
\begin{example}[Linearisation with two templates]
We come back to Example~\ref{strongex}. Recall that $x=(x_1,x_2)$ and $\kappa(x)=\norm{x}_2^2=x_1^2+x_2^2$,
$T^i(x_1,x_2)=(x_1^2+x_2^3,x_1^3+x_2^2)$ and $T^e(x_1,x_2)=(0.5x_1^3+0.4x_2^2,-0.6x_1^2+0.3x_2^2)$. 
Let us consider $\K=\FRP$ and $\pp=\{p_1,p_2\}$. We set $\lambda_1^i(x_1,x_2)=(1,1)$ (in the sense of for the inequality 
involving $p_1$, the weight associated to $p_1$ is 1 and the one associated to $p_2$ is 1), $\lambda_2^i(x_1,x_2)=(1,0)$,
$\lambda_1^e(x_1,x_2)=\lambda_2^e(x_1,x_2)=(1,1)$
and $\nu_{1}(x_1,x_2)=1$ and $\nu_{2}(x_1,x_2)=1$ (for the inequality involving the abstract sublevel 
property, the weights associated to each template are equal to 1). The linearisation consists in the system 
of inequalities:
\[
\left\{
\begin{array}{l}
w_1\geq \sup_{y\in [0.9, 1.1] \times [0, 0.2]} p_1(y)\\
\\
w_2\geq \sup_{y\in [0.9, 1.1] \times [0, 0.2]} p_2(y)\\
\\
\forall x\in\rr^2,\ p_1(x)  +p_2(x) - w_2-p_1(T^i(x))
+\gamma_1^i(x)(\|x\|_2^2-1)\geq 0\\
\\
\forall x\in\rr^2,\ p_1(x) -w_1 +p_2(x)-p_2(T^i(x))
+\gamma_2^i(x)(\|x\|_2^2-1)\geq 0\\
\\
\forall x\in\rr^2,\ p_1(x)  +p_2(x) - w_2-p_1(T^e(x))
+\gamma_1^e(x)(1-\|x\|_2^2)\geq 0\\
\\
\forall x\in\rr^2,\ p_1(x) -w_1 +p_2(x)-p_2(T^e(x))
+\gamma_2^e(x)(1-\|x\|_2^2)\geq 0\\
\\
\forall x\in\rr^2,\ \alpha -\kappa(x) -(w_1-p_1(x))-(w_2-p_2(x))\geq 0
\end{array}
\right.
\]
\end{example}
}\fi

The next lemma states that $\K$ \well templates 
bases are \well template bases. This result is an application of S-Lemma with ``nonnegative 
functions multipliers''.
\begin{lemma}[Functional S-Lemma]
\label{funslemma}
Let $\delta,\beta_1,\ldots,\beta_n \in \rr$ and $h,g_1,\ldots,g_n \in \FR$.
If there exists $\lambda\in \FRP^n$ such that
\begin{equation}
\label{ineq-s-lemma}
\delta-h(x)-\sum_{i=1}^n \lambda_i(x)(\beta_i-g_i(x))\geq 0\enspace ,
\end{equation}
then 
\begin{equation}
\label{impl-s-lemma}
\forall x \in \rd \, , (g_1(x)\leq \beta_1\wedge \ldots\wedge g_n(x)\leq \beta_n \implies h(x)\leq \delta) \, .
\end{equation}
\end{lemma}
\begin{proof}
Assuming that the inequality~\eqref{impl-s-lemma} holds for some $\lambda\in \FRP^n$, we obtain $\delta-h(x)\geq \sum_{i=1}^n \lambda_i(x)(\beta_i-g_i(x))$.
The positivity of $\lambda_i$ yields the desired result.
\qed
\end{proof}
\if{
\begin{corollary}
When $C$ can be represented as $\{x\in\rd\mid g_j(x)\leq 0,\ \forall\, j=1,\ldots,n_C\}$ where 
$g_j\in\FR$. Let $\pp=\{p_1,\ldots,p_k\}$ be a finite template basis. If there exists $\sigma\in \K^{k\times n_C}$
such that for all $l=1,\ldots,k$: 
$x\mapsto w_l-p_l(x)-\sum_{j=1}^{n_C} \sigma_{l,j}(x)g_j(x)\in\K$ then $w_l\geq \sup_{x\in C} p_l(x)$.  
\end{corollary}
\begin{proof}
It suffices to remark that $w_l\geq \sup_{x\in C} p_l(x)$ is equivalent to $g_j(x)\leq 0, j=1,\ldots,n_C \implies 
p_l(x)\leq w_l$ and the result is a direct application of Lemma~\ref{funslemma}.
\qed
\end{proof}
}\fi
\begin{theorem}[$\K$ \well is \well]
\label{soswell}
Assume that a finite template basis $\pp$ is $\K$ \well w.r.t. $\pws$ and $\prop{\kappa}$. Then $\pp$ is \well w.r.t. $\pws$ and $\prop{\kappa}$. 
\end{theorem}
\begin{proof}
$\pp=\{p_1,\ldots,p_k\}$ is $\K$ well-representative. Then there exists 
$w\in\rr^k$, $\alpha\in\rr$ and $\nu\in\K^k$ and for all $i\in\ind$, $\lambda^i\in\K^{k\times k}$, $\mu^i\in\K^{k\times n_i}$, $\gamma^i\in \K^{k\times n_0}$
such that, for all $l=1,\ldots,k$, for all $i\in\ind$, $S_l^i\in\K$, $S^\kappa\in\K\subseteq \FRP$ ($S_l^i\in\K$ and $S^\kappa$ defined at Equation~\eqref{auxiliaryineq}) and 
$w_l\geq \sup\{p_l(x)\mid x \in \xin\}$. We set, for all $l=1,\ldots,k$, $v(p_l):=w_l$.
From Proposition~\ref{Galois}, $v(p_l)\geq \sup\{p_l(x)\mid x \in \xin\}$ for all $l=1,\ldots,k$ is equivalent 
to $\xin\subseteq v^\star$ and $S_l^i\in\K\subseteq \FRP$ for all $l=1,\ldots,k$ and for all $i\in\ind$ imply 
 respectively, by Lemma~\ref{funslemma} 
for all $i\in\ind$, $\left(x\in v^{\star}\wedge r^i(x)\leq 0\wedge r^0(x)\leq 0\implies T^i(x)\in v^{\star}\right)$.
Taking $\clo{v}=(v^\star)^\dag$, we have from Equation~\eqref{galoisprop}, $\clo{v}\in\vep$ and $\clo{v}^\star=v^\star$. 
By Proposition~\ref{inductive}, $v\in\feas{\pws}$. Finally $S^\kappa\in\K\subseteq \FRP$ implies that $v^*\subseteq \{x\in\rd\mid\kappa(x)\leq \alpha\}\subseteq \prop{\kappa}$ by Lemma~\ref{funslemma}. 
\qed
\end{proof}
This proof exhibits a feasible invariant bound which is given by the variable $w$ of the system of
inequalities in Definition~\ref{kwell}.  
\subsection{Simple construction of $\K$ well-representative template bases}
\label{simplemethods}
In this subsection, we discuss how to simply construct $\K$ well-representative template bases. 
\begin{proposition}[With one $\K$ well-representative template]
\label{onlyone}
Let $\{p\}$ be a $\K$ \well template basis w.r.t. $\pws$ and $\prop{\kappa}$ and 
$\mathcal{Q}$ be a finite subset of $\FR$ $\st$ for all $q\in\mathcal{Q}$, $p-q\in\K$, for all $i\in\ind$, $(p-q)\circ T^i\in\K$. Then $\pp=\{p\}\cup \mathcal{Q}$ is a $\K$ \well template basis w.r.t. $\pws$ and 
$\prop{\kappa}$.  
\end{proposition}
\begin{proof}
Suppose that $\{p\}$ is $\K$ \well w.r.t. $\prop{\kappa}$. 
By definition, there exists $w\in\rr$, $\alpha\in\rr$ and $\nu\in\K$ and for all $i\in\ind$, $\lambda^i\in\K$, $\mu^i\in \K^{1\times n_i}$, $\gamma^i,\in\K^{1\times n_0}$, $\nu\in\K$ such that the functions for all $i\in\ind$, $S^i := S_1^i$, $S^\kappa$
belong to $\K$ ($S_1^i\in\K$ and $S^\kappa$ defined at Equation~\eqref{auxiliaryineq}) and $w\geq \sup\{p(x)\mid x\in \xin\}$. Let us take $q$ such that $p-q\in\K$. It follows that $p\geq q$ and thus:
$w\geq \sup\{p(x)\mid x\in \xin\} \geq \sup\{q(x)\mid x\in \xin\}$. Now let $i\in\ind$, since $(p-q)\circ T^i\in \K$ then there exists $f\in\K$ such that $f(x)=p(T^i(x))-q(T^i(x))$ for all $x\in\rd$, we have $w(1-\lambda^i(x))-q(T(x))+\lambda^i(x)p(x)+\sum_{j=1}^{n_i}\mu_j^i(x) 
r_j^i(x)+\sum_{j=1}^{n_0}\gamma_j^i(x)r_j^0(x)=S^i(x)+ f(x)$ for all $x\in\rd$. Since $\K$ is closed under addition then $S^i+ f\in \K$. Now $S^\kappa\in\K$ implies that $S^\kappa+0(w-q)\in\K$.
It follows that $\{p,q\}$ is $\K$ \well w.r.t. $\pws$ and $\prop{\kappa}$ by taking
$(w,w)\in\rr^2$, $\alpha\in\rr$, $(\nu,0)\in\K^2$ and for all $i\in\ind$, $\{(\lambda^i,0),(\lambda^i,0)\}\in\K^{2\times 2}$, 
$(\mu^i,\mu^i)\in\K^{2\times n_i},(\gamma^i,\gamma^i)\in\K^{2\times n_0}$ (following
the order of the parameters of Definition~\ref{kwell}). We conclude by induction on the elements $q$.\qed
\end{proof}

\begin{example}[With Quadratic Lyapunov Functions]
\label{lyapunov}
Let us consider the following program:
\begin{center}
\begin{tabular}{c}
\begin{lstlisting}[mathescape=true]
x $\in$ $\xin$;
while (-1<=0){
   x = $A$x;
}
\end{lstlisting}
\end{tabular}
\end{center}
where $\xin$ is a bounded set, $A$ is a $d\times d$ matrix. Its CPDS representation is $S=(\xin,\rd,\rd,Ax)$.
Suppose there exists a symmetric matrix $P$ such that: 
\begin{equation}
\label{eqlyap}
P-\Id\succeq 0\qquad P-A^\intercal P A\succeq 0
\end{equation}
 where $B-C\succeq 0$ for two symmetric matrices means that $x^\intercal (B-C) x\geq 0$ 
for all $x$ and $\Id$ is the identity matrix. Let $k=1,\ldots d$ and let us denote by $I_k$ the $d\times d$ matrix
such that $I_k(i,j)=1$ if $i=j=k$ and 0 otherwise. Remark that $\Id-I_k\succeq 0$ for all $k=1,\ldots d$.
  
Let $\K=\{x\mapsto x^\intercal Q x+c\mid c\in\rr_+,Q\succeq 0\}$.
Then $\pp=\{x\mapsto x^\intercal P x\}\cup\{x\mapsto x^\intercal I_k x, k=1,\ldots,d\}$ 
is a $\K$ well-representative template basis w.r.t. $S$ and $\prop{\|\cdot\|_2^2}$.

We write $\beta:=\sup\{x^\intercal Px \mid x\in \xin\}\in\rr$ (since $\xin$ is bounded and $x\mapsto x^\intercal P x$
is continuous). We have to exhibit $w,\alpha\in\rr$ and $\lambda,\nu\in\K$ such that:  
$w\geq \beta$, $x\mapsto w-\lambda(x)(w-x^\intercal P x)-x^\intercal A^\intercal PAx\in \K$
and $x\mapsto \alpha -\norm{x}_2^2-\nu(x)(w-x^\intercal P x)\in \K$.
Taking $\lambda=\nu=1$ and $\alpha=w=\beta$, the latter inequalities become  
$P-A^\intercal PAx\succeq 0$ and $x\mapsto-\norm{x}_2^2+x^\intercal P x\geq 0$. So $-\norm{x}_2^2+x^\intercal P x=x^\intercal (P-\Id) x\in\K$. Thus, 
$\{x\mapsto x^\intercal P x\}$ is a $\K$ well-representative template basis w.r.t. $S$ and $\prop{\|\cdot\|_2^2}$. 
Now $P-Id\succeq 0$ implies that $P-I_k\succeq 0$ and then $x^\intercal P x-x^\intercal I_k x\in\K$. For all $k=1,\ldots d$,
for all $x\in\rd$, $x^\intercal A^\intercal P Ax-x^\intercal A^\intercal I_k A x=x^\intercal A^\intercal P-I_k A x\in\K$.
By Proposition~\ref{onlyone}, a $\K$ well-representative template basis w.r.t. $S$ and $\prop{\|\cdot\|_2^2}$.

This example shows that the quadratic forms (Lyapunov functions for discrete-time linear systems) $x\mapsto x^\intercal P x$ for $P$ satisfying Equation~\eqref{eqlyap} combined with $x\mapsto x_k^2$ are used in the setting of quadratic templates.
\end{example}
Another possibility consists in 
constructing a $\K$ \well template basis w.r.t. $\pws$ and $\prop{\kappa}$ from a vector of 
templates $p_1, \dots, p_k$ such that for all $i = 1, \dots, k$, $\{p_i\}$ is a $\K$ \well templates w.r.t. $\pws$ and $\prop{\kappa}$ (Proposition~\ref{twowell}).
\begin{proposition}[From two single $\K$ well-representative templates]
\label{twowell}
Let $\{p\}$ and $\mathcal{Q}$ two $\K$ well-representative template bases w.r.t. $\pws$ and $\prop{\kappa}$.
Then $\{p\}\cup \mathcal{Q}$ is a $\K$ well-representative template basis w.r.t. $\pws$ and $\prop{\kappa}$.
\end{proposition}
\begin{proof}
By induction, it suffices to prove the result for $\mathcal{Q}=\{q\}$. We write
$p_1=p$ and $p_2=q$. By definition, for $l=1,2$, there exist $w_l\in\rr$, $\alpha_l\in\rr$, $\nu\in\K^k$ and for all $i\in\ind$ $\lambda_l^i\in \K,\mu_l^i\in\K^{1\times n_i},\gamma_l^i\in\K^{1\times n_0}$ such that
$S_{l}^i, S_{l}^\kappa\in\K$ ($S_l^i\in\K$ and $S^\kappa$ defined at Equation~\eqref{auxiliaryineq}) and 
$w_l\geq \sup\{p_l(x)\mid x\in \xin\}$. It follows that $\{p,q\}$ is $\K$ \well w.r.t. $\pws$ and $\prop{\kappa}$ by taking
$(w_1,w_2)\in\rr^2$, $\alpha=(\alpha_1+\alpha_2)/2\in\rr$, ($\nu_1/2,\nu_2/2)\in\K^2$ and for all $i\in\ind$, $\{(\lambda_1^i,0),(0,\lambda_2^i)\}\in\K^{2\times 2}$, 
$(\mu_1^i,\mu_2^i)\in\K^{2\times n_i},(\gamma_1^i,\gamma_2^i)\in\K^{2\times n_0}$ (following
the order of the parameters in Definition~\ref{kwell}). 
To conclude, we use the fact that $\K$ is closed under nonnegative scalar multiplications. \qed
\end{proof} 
\subsection{Practical computation using sum-of-squares programming}
\label{sos}
Let $\rr[x]$ stands for the set of $d$-variate polynomials and $\rr_{2m}[x]$ be its subspace of polynomials of degree at most $2 m$. 
We instantiate $\mathcal{K}$ by the cone of sum-of-squares (SOS), that is
$\mathcal{K} = \Sigma[x] := \Bigl\{\,\sum_i q_i^2, \, \text{ with } q_i \in \rr[x] \Bigr\}$.

In the sequel, we assume that the data of the CPDS representation $\pws$ of some analyzed program  are polynomials, that is for all $j=1,\ldots,n_0$, $r_j^{\mathrm{in}}\in\rr[x]$, for all $j=1,\ldots,n_0$, $r_j^0\in\rr[x]$, for all $i\in\ind$, $T^i\in\rr[x]$ and for all $j=1,\ldots, n_i$, $r_j^i\in \rr[x]$. We look for a single polynomial template $p \in \rr_{2m}[x]$ ($k = 1$) such that the basis $\{p\}$ is $\Sigma[x]$ \well w.r.t. $\pws$ and $\prop{\kappa}$, thus satisfies the three conditions of Definition~\ref{kwell}.
One way to strengthen the three conditions of Definition~\ref{kwell} is to take  $\lambda^i = 1$, for all $i\in\ind$, $\nu = 1, \alpha = w$, then to consider the following {\em hierarchy} of SOS constraints, parametrized by the integer $m$:
\begin{equation}
\label{eq:singlesos}
\left\{
\begin{array}{l}
\displaystyle{w -p(x)+\sum_{j=1}^{n_{\mathrm{in}}} \sigma_j(x) r_j^{\mathrm{in}}(x) = \sigma_0(x)} \enspace , \\
\forall\, i\in\ind,\ \displaystyle{-p(T^i(x))+ p(x)+\sum_{j=1}^{n_i}\mu_j^i(x) r_j^i(x)+\sum_{j=1}^{n_0}\gamma_j^i(x) r_j^0(x) = \sigma^i(x)}  \enspace , \\
\displaystyle{-\kappa(x) + p(x) = \psi(x) } \enspace , \\
\\
p \in \rr_{2 m}[x] \enspace, w \in \rr \enspace, \\
\sigma_0 \in \Sigma[x] \enspace , \ \deg \sigma_0 \leq 2 m \enspace ,\\
\forall\, j = 1, \dots, n_{\mathrm{in}}\enspace ,\ \sigma_j \in \Sigma[x] \enspace ,\ \deg (\sigma_j g_j) \leq 2 m  \enspace, \\
\\
\forall\, i\in\ind \enspace ,\ \sigma^i\in \Sigma[x]\enspace ,\ \deg (\sigma^i) \leq 2 m \deg T^i \enspace,\\
\forall\, i\in\ind \enspace ,\ \forall\, j=1,\ldots, n_i \enspace,\ \mu_j^i\in\Sigma[x]\enspace ,\ \deg (\mu_j^i r_j^i)  \leq 2 m \deg T^i \enspace ,\\
\forall\, i\in\ind\enspace ,\ \forall\, j=1,\ldots, n_0\enspace ,\ \gamma^i\in \Sigma[x]\enspace ,\ \deg (\gamma_j^i r_j^0)  \leq 2 m \deg T^i  \enspace ,\\
\\
\psi \in \Sigma[x] \enspace,\ \deg (\psi) \leq 2 m  \enspace .\\
\end{array}
\right.
\end{equation}
%
For an integer $m$, we denote by $\cont_m$ the set of constraints on the decision variables $w,p,\sigma^0,\{\sigma_j, j= 1, \dots, n_{\mathrm{in}}\}, \{\sigma^i, i\in\ind\},\{\mu_j^i,i\in\ind,j=1,\ldots,n_i\},\{\gamma_j^i,i\in\ind,j=1,\ldots,n_0\}$ and $\psi$ depicted at Equation~\eqref{eq:singlesos}. 

As objective function, we choose to minimize $w$. The intuition behind this choice is that $w$ is enforced to be equal to $\alpha$ which defines the level for which $\{x\in\rd\mid \kappa(x)\leq \alpha\}$ is an invariant of the program associated to the CPDS $\pws$. When $\kappa$ is the norm, a minimal value $w$ (and thus $\alpha$) would be the smallest computable bound on the norm of the state variable $x_k$. Thus we synthetize a polynomial template of degree at most $2m$ by solving the following minimization problem:
\begin{equation}
\label{polsynthesis}
\inf\left\{ w\in\rr \left| 
\left(\begin{array}{l}
w,p,\sigma^0,\{\sigma_j, j= 1, \dots, n_{\mathrm{in}}\},\\
\{\sigma^i, i\in\ind\},\{\mu_j^i,i\in\ind,j=1,\ldots,n_i\},\\
 \{\gamma_j^i,i\in\ind,j=1,\ldots,n_0\},\psi) 
\end{array}\right)\in\cont_m \right\}\right. \,.
\end{equation}
        
Hence, computing the polynomial template $p \in \rr_{2m}[x]$ boils down to solving an SOS minimization problem. 
 From an optimal solution of Program~\eqref{polsynthesis}, one can extract the polynomials $\sigma_0, \sigma_1, \dots, \sigma_{n_{\mathrm{in}}}, \psi \in \Sigma[x]$ and for all $i\in\ind$, the polynomials $\mu^i, \gamma^i, \sigma^i\in\Sigma[x]$, which are called {\em SOS certificates}.  In practice, one can use the Matlab toolbox {\sc Yalmip} \cite{YALMIP}, which includes a high-level parser for nonlinear optimization and has a built-in module for such SOS calculations. {\sc Yalmip} reduces SOS programming to semidefinite programming (SDP) (see e.g.~\cite{Vandenberghe94semidefiniteprogramming} for more details about SDP), which in turn can be handled with efficient SDP solvers, such as {\sc Mosek}~\cite{mosek}.
In our setting the choice $\alpha = w$ avoids numerical issues while solving SDP programs.
\if{
\begin{definition}[SOS \well template basis]
A template basis $\pp=\{p_1,\ldots,p_k\}$ is SOS \well template basis w.r.t. 
$Pr(C,r,s,T^i,T^e)$ iff $\pp$ is composed by polynomials, there exist $w \in\rr^k$, $\alpha \in \rr^{n_\kappa}$
$\lambda^i,\lambda^e\in \Sigma_m[x]^{k\times k}$, $\mu^i,\mu^e,\gamma^i,\gamma^e\in \Sigma_m[x]^{k}$, 
$\sigma\in \Sigma_m[x]^{k\times n}$ and $\nu\in \Sigma_m[x]^{t\times k}$ such that:
\[
\forall \, l=1,\ldots k, \, w_l-p_l(x)+\sum_{j=1}^{n_C} \sigma_{l,j}(x) g_j(x) \in \Sigma_m[x]
\] 
and
\[
\forall \, l=1,\ldots k, \, w_l -\sum_{j=1}^k \lambda_{l,j}^i(x) (w_j-p_j(x))-p_l(T^i(x))+\mu_l^i(x)r(x)+\gamma_l^i(x)s(x)\in \Sigma_m[x]
\]
and
\[
\forall \, l=1,\ldots k, \, w_l -\sum_{j=1}^k \lambda_{l,j}^e(x) (w_j-p_j(x))-p_l(T^e(x))+\mu_l^e(x)r(x) - \gamma_l^e(x)s(x)\in \Sigma_m[x]
\]
and 
\[\alpha -\kappa(x) -\sum_{t=1}^k \nu_{t}(x) (w_t-p_t(x))\in \Sigma_m[x].\]
\end{definition}
}\fi

\if{
\begin{problem}
\label{eqtemplatefinalSOSk}
For a fixed positive integer $m$, find $P \in \rr_{2 m}[x], \gamma_2^i\in\rr_+$ such that:
\begin{enumerate}
\item  $P -\norm{\cdot}_q^q \in \Sigma_{m}[x]$ ;
\item $\forall\, i\in\affect,\ P-P \circ T^i \in \Sigma_{m}[x]$ ;
\item $\forall\, i\in\inter,\ P - P \circ T^i +\gamma_2^i r \in \Sigma_{m}[x]$ ;
\end{enumerate}
with $k := \max \{\bigl\lceil \dfrac{q}{2} \bigr\rceil, \bigl\lceil \dfrac{\deg r}{2} \bigr\rceil,  \max_{i \in \affect \cup  \inter} \{ k' \bigl\lceil \dfrac{\deg T^i}{2} \bigr\rceil \} \}$.
\end{problem}
}\fi
\paragraph*{Computational considerations} Define $t := \max\{\deg T^i, i\in\ind \}$. At step $m$ of this hierarchy, the number of SDP variables is proportional to $\binom{d + 2 m t}{d}$ and the number of SDP constraints is proportional to $\binom{d + m t}{d}$. Thus, one expects tractable approximations when the number $d$ of variables (resp. the degree $2 m$ of the template $p$) is small. However, one can handle bigger instances of Problem~\eqref{polsynthesis} by taking into account the system properties. For instance one could exploit sparsity as in~\cite{Waki06sumsof} by considering the variable sparsity correlation pattern of the polynomials $\{T^i,i\in\ind\},\{r_j^i, i\in\ind, j=1,\ldots,n_i\},\{r_j^0,j=1,\ldots,n_0\},\{r_j^{\mathrm{in}},j=1,\ldots,n_{\mathrm{in}}\}$ and $\kappa$. 

Recall that $\rea(\pws)$ is the set of possible values taken by the CPDS $\pws$, which are also the possible values taken by the variables of the program represented by $\pws$.
\begin{proposition}
\label{outer}
Assume that step $m$ of Problem~\eqref{polsynthesis} yields a feasible solution and denote by $p^{(m)} \in \rr_{2m}[x]$ (resp. $w^{(m)}$) the polynomial  template (resp. the upper bound of $p^{(m)}$ over $\xin$) associated to this solution. Let $v(p^{(m)})=w^{(m)}$ and thus $v^\star := \{x \in \rr^d \mid p^{(m)}(x) \leq w^{(m)}\}$. 
Then $\rea(\pws) \subseteq v^\star$ and $x\in v^\star\implies \kappa(x)\leq w^{(m)}$.
\end{proposition}
\begin{proof}
As a consequence of the first equality constraint of Problem~\eqref{eq:singlesos}, one has $w^{(m)} \geq \sup_{x \in \xin} p^{(m)}(x)$. Then, the finite template basis $\{p^{(m)}\}$ is $\Sigma[x]$ well-representative w.r.t. $\pws$ and $\prop{\kappa}$. 
By Theorem~\ref{soswell}, this basis is well-representative w.r.t. $\pws$ and $\prop{\kappa}$. In the proof of Theorem~\ref{soswell}, we also proved that $v\in\feas{\pws}$ and $v^\star\subseteq \prop{\kappa}$. Thus from the second statement of Proposition~\ref{basicresults}, $\rea(\pws) \subseteq v^\star$ and $\prop{\kappa}$ is sublevel invariant.   
\qed
\end{proof}
The next corollary follows directly from  Proposition~\ref{twowell} and Proposition~\ref{outer}.
\begin{corollary}
\label{outerbasis}
Given some integers $k$ and $m$, assume that steps $m, \dots, m + k$ of Problem~\eqref{polsynthesis} yield respective feasible polynomial solutions $p^{(m)}, \dots, p^{(m + k)}$. Then, $\{ p^{(m)}, \dots, p^{(m + k)} \}$ is a $\Sigma[x]$ \well template basis w.r.t. $\pws$ and $\prop{\kappa}$.
\end{corollary}

\section{Benchmarks}
Here, we perform some numerical experiments while solving Problem~\eqref{polsynthesis} (given in Section~\ref{sos}) on several examples. 
In Section~\ref{benchbound}, we verify that the program of Example~\ref{running} satisfies some boundedness property. We also provide examples involving higher dimensional cases. Then, Section~\ref{benchsafe} focuses on checking that the set of variable values avoids an unsafe region. Numerical experiments are performed on an Intel Core i5 CPU ($2.40\, $GHz) with {\sc Yalmip} being interfaced with the SDP solver {\sc Mosek}. For the sake of simplicity, we write $w_m^{\star}$ instead of $v(p^{(m)})=w^{(m)}$.
\label{bench}
\subsection{Checking boundedness of the set of variables values}
\label{benchbound}
\begin{example}
\label{ex:test}
Following Example~\ref{running}, we consider the constrained piecewise discrete-time dynamical system $\pws=(\xin,X^0,\{X^1,X^2\},\{T^1,T^2\})$  with $\xin = [0.9, 1.1] \times [0, 0.2] $, $X^0=\{x\in\rr^2\mid r^0(x)\leq 0\}$ with $r^0:x\mapsto -1$, 
$X^1=\{x\in\rr^2\mid r^1(x)\leq 0\}$ with $r^1:x\mapsto \norm{x}^2-1$, $X^2=\{x\in\rr^2\mid r^2(x)<0\}$
with $r^2=-r^1$ and $T^1:(x_1,x_2)\mapsto (c_{11}x_1^2+c_{12}x_2^3,c_{21}x_1^3+c_{22}x_2^2)$, 
 $T^2:(x_1,x_2) \mapsto (d_{11}x_1^3+d_{12}x_2^2,d_{21}x_1^2+d_{22}x_2^2)$. We are interested in proving the boundedness property which a sublevel property $\prop{\kappa}$ with $\kappa : x \mapsto \| x \|_2^2$. 
\end{example}
\begin{figure}[!ht]
\centering
\subfigure[$m = 3$]{
\includegraphics[scale=\sizesmallfig]{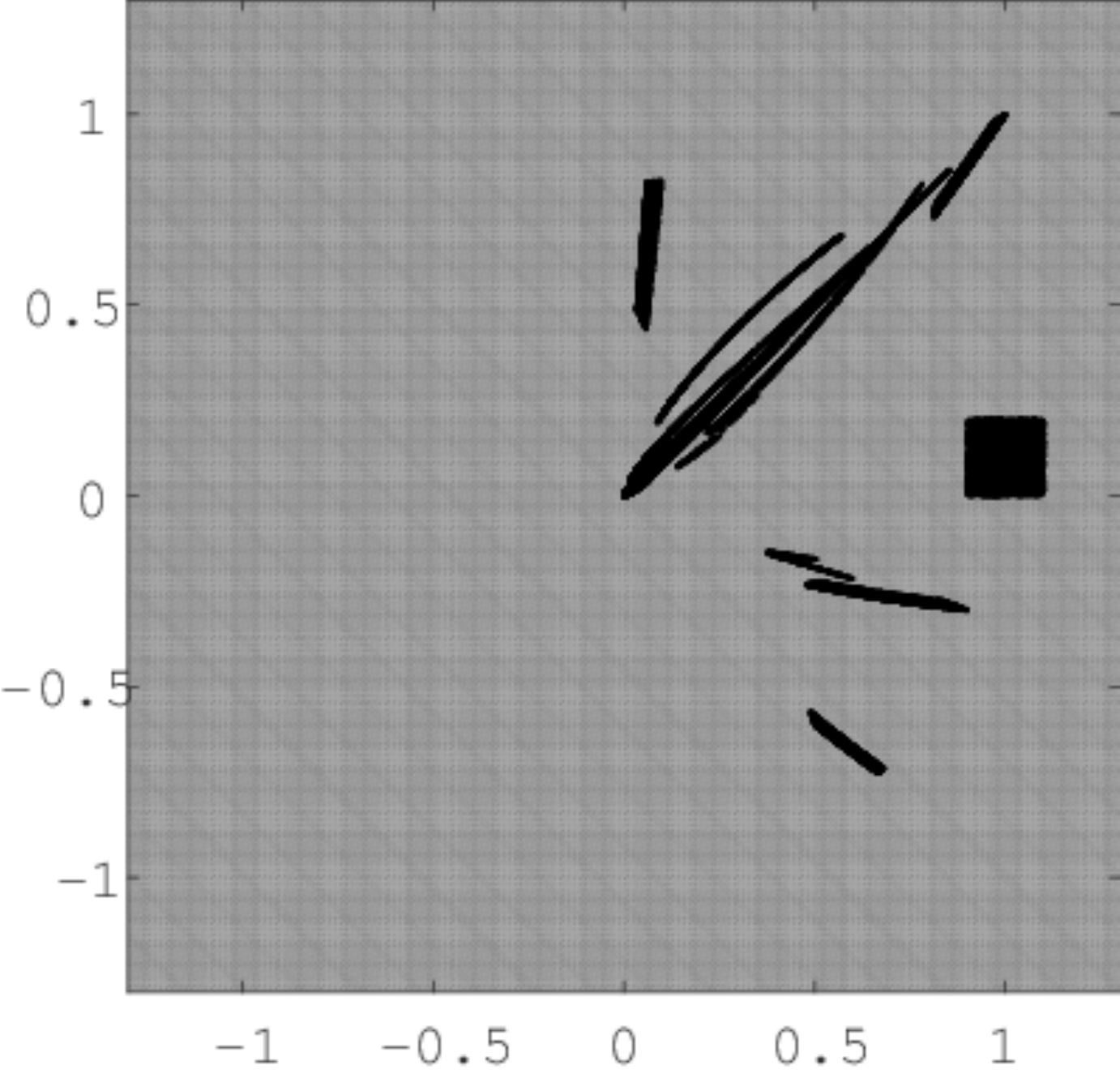}}
\subfigure[$m = 4$]{
\includegraphics[scale=\sizesmallfig]{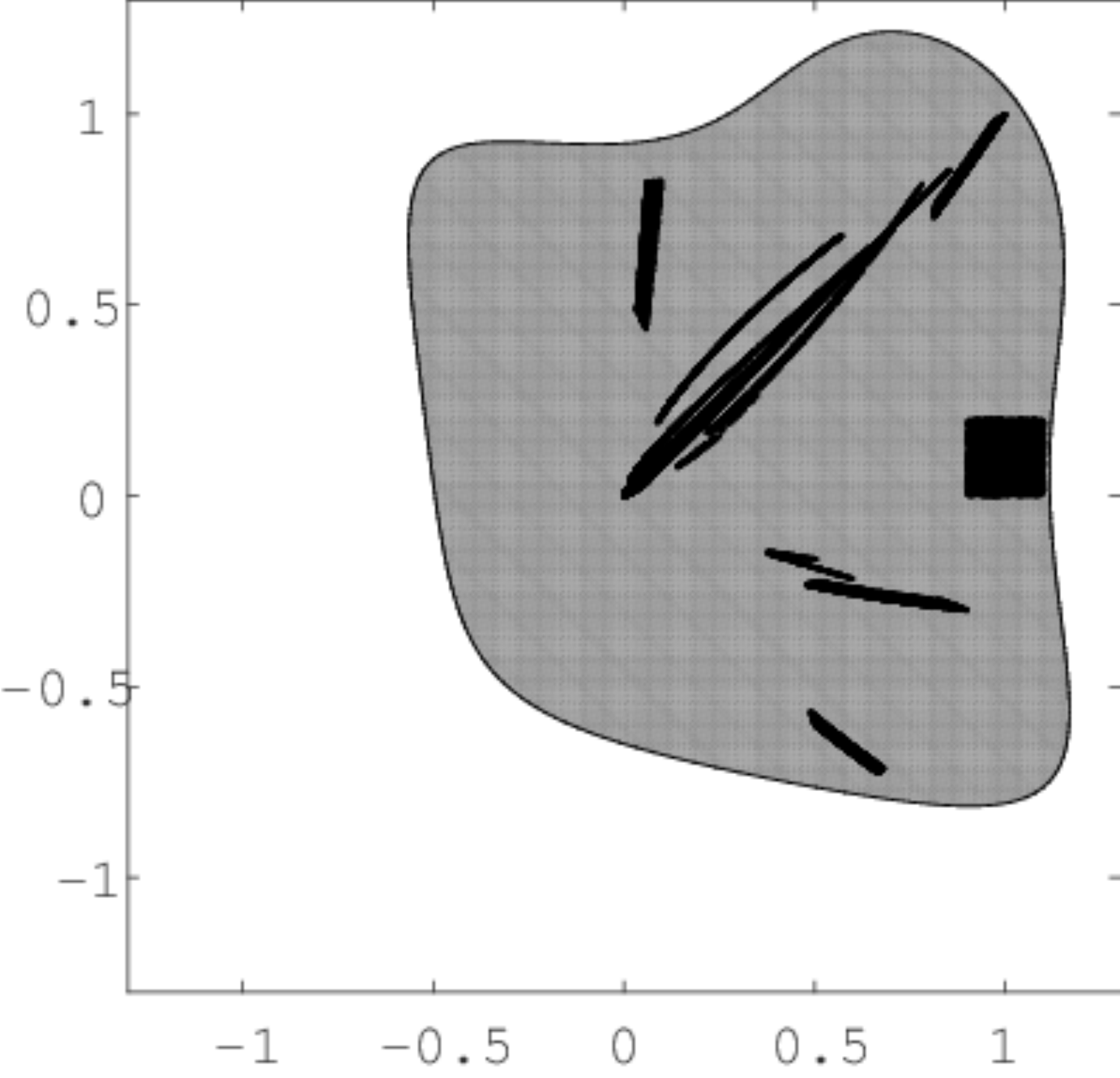}}
\subfigure[$m = 5$]{
\includegraphics[scale=\sizesmallfig]{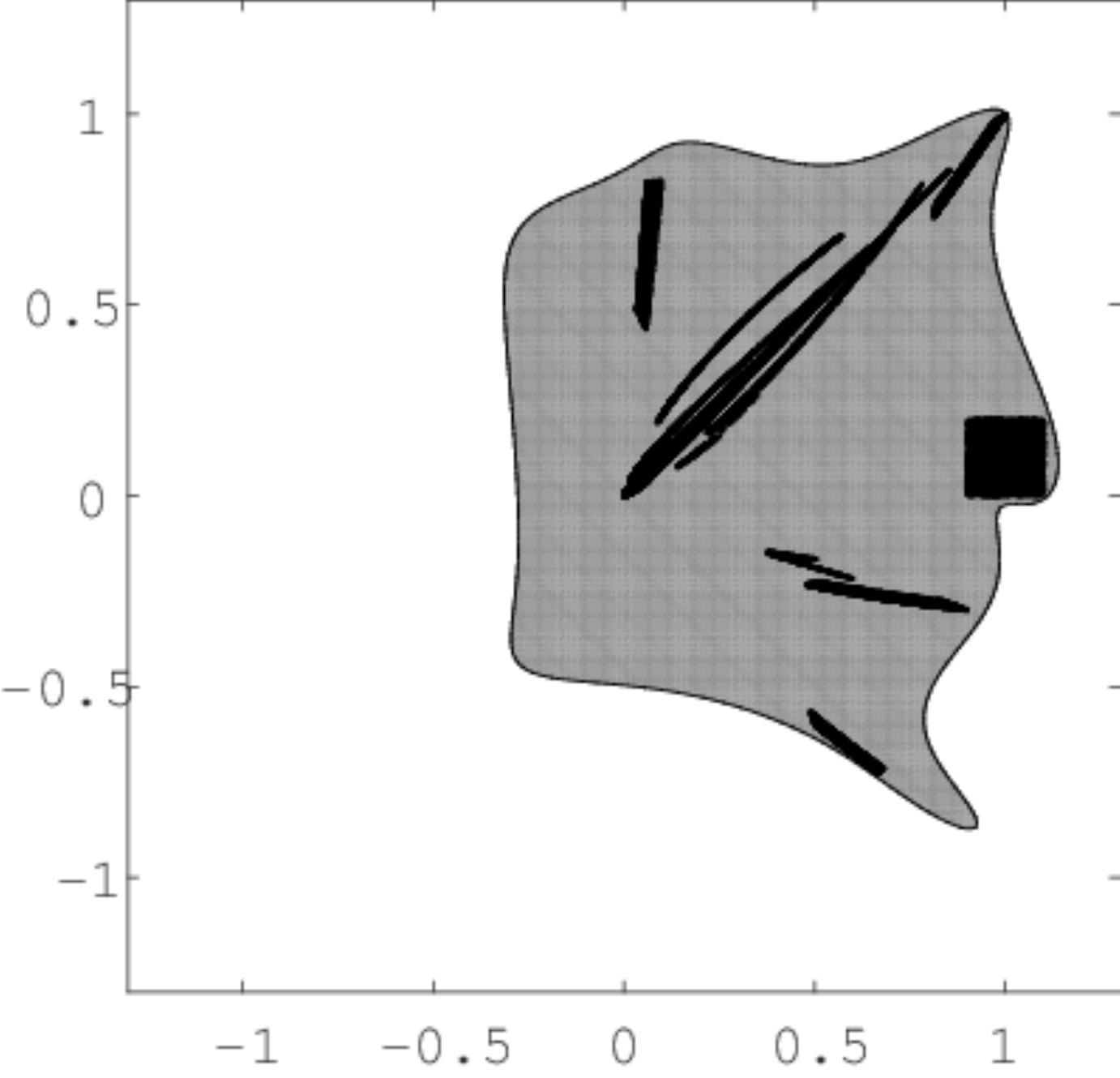}}
\caption{A hierarchy of sublevel sets $w_m^{\star}$ for Example~\ref{ex:test}}	
\label{fig:test}
\end{figure}
Here we illustrate the method by instantiating the program of Example~\ref{running} with the following input: $a_1 = 0.9$, $a_2 = 1.1$, $b_1 = 0$, $b_2 = 0.2$, $c_{11} = c_{12} = c_{21} = c_{22} = 1$,  $d_{11} = 0.5$,  $d_{12} = 0.4$, $d_{21} = -0.6$ and $d_{22} = 0.3$.
We represent the possible initial values taken by the program variables $(x_1, x_2)$ by picking uniformly $N$ points $(x_1^{(i)}, x_2^{(i)}) \ (i = 1, \dots, N)$  inside the box $\xin = [0.9, 1.1] \times [0, 0.2]$ (see the corresponding square of dots on Figure~\ref{fig:test}). The other dots are obtained after successive updates of each point $(x_1^{(i)}, x_2^{(i)})$ by the program of Example~\ref{running}. The sets of dots in Figure~\ref{fig:test} are obtained with $N = 100$ and six successive iterations.

At step $m = 3$, Program~\eqref{polsynthesis} already yields a feasible solution, from which one can extract the polynomial template $p^{(3)}$ and $w_3 \in \feas{\pws}$.  The SOS certificates extracted from this solution guarantee the boundedness property, that is $x \in \rea(\pws)\implies x \in w_3^\star \implies \| x \|_2^2 \leq w^{(3)}$. 
Figure~\ref{fig:test} displays in light gray outer approximations of the set of possible values $X_1$ taken by the program of Example~\ref{ex:test} as follows: (a) the degree six sublevel set $w_3^\star$, (b) the degree eight sublevel set $w_4^\star$ and (c) the degree ten sublevel 
set $w_5^\star$. The outer approximation $w_3^\star$ is coarse as it contains the box $[-1.5, 1.5]^2$. However, solving Problem~\eqref{polsynthesis} at higher steps yields tighter outer approximations of $\rea(\pws)$ together with more precise bounds $w^{(4)}$ and $w^{(5)}$. Finally, $\{p^{(3)}, p^{(4)}, p^{(5)}\}$ is a $\Sigma[x]$ well-representative template basis w.r.t. to $\pws$ and $\prop{\|\cdot\|_2^2}$ for the program of Example~\ref{ex:test}.

We also succeeded to certify that the same property holds for higher dimensional programs, described in Example~\ref{ex:test3} ($d = 3$) and Example~\ref{ex:test4} ($d = 4$).
\begin{example}
\label{ex:test3}
Here we consider $\xin = [0.9, 1.1] \times [0, 0.2]^2 $, $r^0:x\mapsto -1$, 
$r^1: x \mapsto \| x \|_2^2-1$, $r^2=-r^1$, $T^1:(x_1,x_2,x_3)\mapsto 1/4 (0.8 x_1^2 + 1.4 x_2 - 0.5  x_3^2,   1.3  x_1 + 0.5  x_3^2,  1.4  x_2 + 0.8  x_3^2)$, 
 $T^2:(x_1,x_2, x_3) \mapsto 1/4(0.5  x_1 + 0.4  x_2^2,   -0.6  x_2^2 + 0.3  x_3^2,  0.5  x_3 + 0.4  x_1^2)$ and $\kappa : x \mapsto \| x \|_2^2$.
\end{example}
\begin{example}
\label{ex:test4}
Here we consider $\xin= [0.9, 1.1] \times [0, 0.2]^3 $, $r^0:x\mapsto -1$, 
$r^1: x \mapsto \| x \|_2^2-1$, $r^2=-r^1$, $T^1:(x_1,x_2,x_3, x_4)\mapsto 0.25 (0.8  x_1^2 + 1.4 x_2 - 0.5  x_3^2, 1.3  x_1 + 0.5,  x_2^2 - 0.8  x_4^2,         0.8  x_3^2 + 1.4 x_4,    1.3  x_3 + 0.5  x_4^2)$, 
 $T^2:(x_1,x_2, x_3, x_4) \mapsto 0.25 (0.5  x_1 + 0.4  x_2^2,   -0.6  x_1^2 + 0.3  x_2^2, 0.5  x_3 + 0.4  x_4^2,   -0.6  x_3 + 0.3  x_4^2)$ and $\kappa : x \mapsto \| x \|_2^2$.
\end{example}
Table~\ref{table:bench} reports several data obtained while solving Problem~\eqref{polsynthesis} at step $m$, ($2 \leq m \leq 5$), either for Example~\ref{ex:test}, Example~\ref{ex:test3} or Example~\ref{ex:test4}.
Each instance of Problem~\eqref{polsynthesis} is recast as an SDP program, involving a total number of ``Nb. vars'' SDP variables, with an SDP matrix of size ``Mat. size''. We indicate the CPU time required to compute the optimal solution of each SDP program with {\sc Mosek}.

The symbol ``$-$'' means that the corresponding SOS program could not be solved within one day of computation. These benchmarks illustrate the computational considerations mentioned in Section~\ref{sos} as it takes more CPU time to analyze higher dimensional programs.
Note that it is not possible to solve Problem~\eqref{polsynthesis} at step $5$ for Example~\ref{ex:test4}. A possible workaround to limit this computational blow-up would be to exploit the sparsity of the system.

%
\renewcommand{\tabcolsep}{0.4cm}
\begin{table}[!ht]
\begin{center}
\caption{Comparison of timing results for Example~\ref{ex:test},~\ref{ex:test3} and~\ref{ex:test4}}
\begin{tabular}{c|c|cccc}
\hline
\multicolumn{2}{c|}{Degree $2 m$}
& 4 & 6 & 8 & 10
\\
\hline  
\multirow{2}{*}{Example \ref{ex:test}} & Nb. vars &  1513 & 5740 & 15705 & 35212 \\
& Mat. size & 368 & 802 & 1404 & 2174 \\
 ($d = 2$) & Time & $0.82 \, s$ & $1.35 \, s$ &  $4.00 \, s$ & $9.86 \, s$\\
\hline
\multirow{2}{*}{Example \ref{ex:test3}} & Nb. vars &  2115 & 11950 & 46461 & 141612\\
& Mat. size & 628 & 1860 & 4132 & 7764 \\
 ($d = 3$)& Time & $0.84 \, s$ & $2.98 \, s$ &  $21.4 \, s$ & $109 \, s$\\
\hline
\multirow{2}{*}{Example \ref{ex:test4}} & Nb. vars & 7202  & 65306 & 18480 & $-$\\
    & Mat. size & 1670 & 6622 & 373057 & $-$\\
 ($d = 4$) & Time & $2.85 \, s$ & $57.3 \, s$ &  $1534 \, s$ & $-$\\
\hline
\end{tabular}
\label{table:bench}
\end{center}
\end{table}

\subsection{Avoiding unsafe regions for the set of variables values}
\label{benchsafe}
Here we consider the program given in Example~\ref{ex:test4}. One is interested in showing that the set $X_1$ of possible values taken by the variables of this program does not meet the ball $B$ of center $(-0.5, -0.5)$ and radius $0.5$.
\begin{example}
\label{ex:testout}
Let consider the CPDS $\pws=(\xin,X^0,\{X^1,X^2\},\{T^1,T^2\})$ with $\xin = [0.5, 0.7] \times [0.5, 0.7] $, $X^0=\{x\in\rr^2\mid r^0(x)\leq 0\}$ with $r^0:x\mapsto -1$, $X^1=\{x\in\rr^2\mid r^1(x)\leq 0\}$ with $r^1:x\mapsto \| x \|_2^2-1$, $X^2=\{x\in\rr^2\mid r^2(x)\leq 0\}$ with $r^2=-r^1$ and $T^1:(x_1,x_2)\mapsto (x_1^2+ x_2^3,  x_1^3+  x_2^2)$, 
 $T^2:(x,y) \mapsto (0.5 x_1^3 + 0.4 x_2^2, - 0.6 x_1^2 + 0.3 x_2^2)$. With $\kappa : (x_1, x_2) \mapsto 0.25 - (x_1 + 0.5)^2 - (x_2 + 0.5)^2$, one has $B := \{ x \in \rr^2 \mid  0 \leq \kappa(x) \}$ and one shall prove that $x \in \rea(\pws) \implies \kappa(x) < 0$. Note that $\kappa$ is not a norm, by contrast with the previous examples.
\end{example}
At steps $m = 3, 4$, Program~\eqref{polsynthesis} yields feasible solutions with nonnegative bounds $w^{(3)}, w^{(4)}$. Hence, it does not allow to certify that $\rea(\pws) \cap B$ is empty. This is illustrated in both Figure~\ref{fig:testout} (a) and Figure~\ref{fig:testout} (b), where the light grey region does not avoid the ball $B$. However, solving the SOS feasibility program at step $m = 5$ yields a negative bound $w^{(5)}$ together with a certificate that $\rea(\pws)$ avoids the ball $B$ (see Figure~\ref{fig:testout} (c)). Finally, $\{p^{(5)}\}$ is a single polynomial template basis w.r.t. $\pws$ and $\prop{\kappa}$ with the restriction that $\{x\in\rd\mid p^{(5)}(x)\leq w^{(5)}\}\subseteq \{x\in\rd\mid \kappa(x)\leq \alpha\}$ for some $\alpha<0$ for the program of Example~\ref{ex:testout}.

%
\begin{figure}[!ht]
\centering
\subfigure[$m=3$]{
\includegraphics[scale=\sizesmallfig]{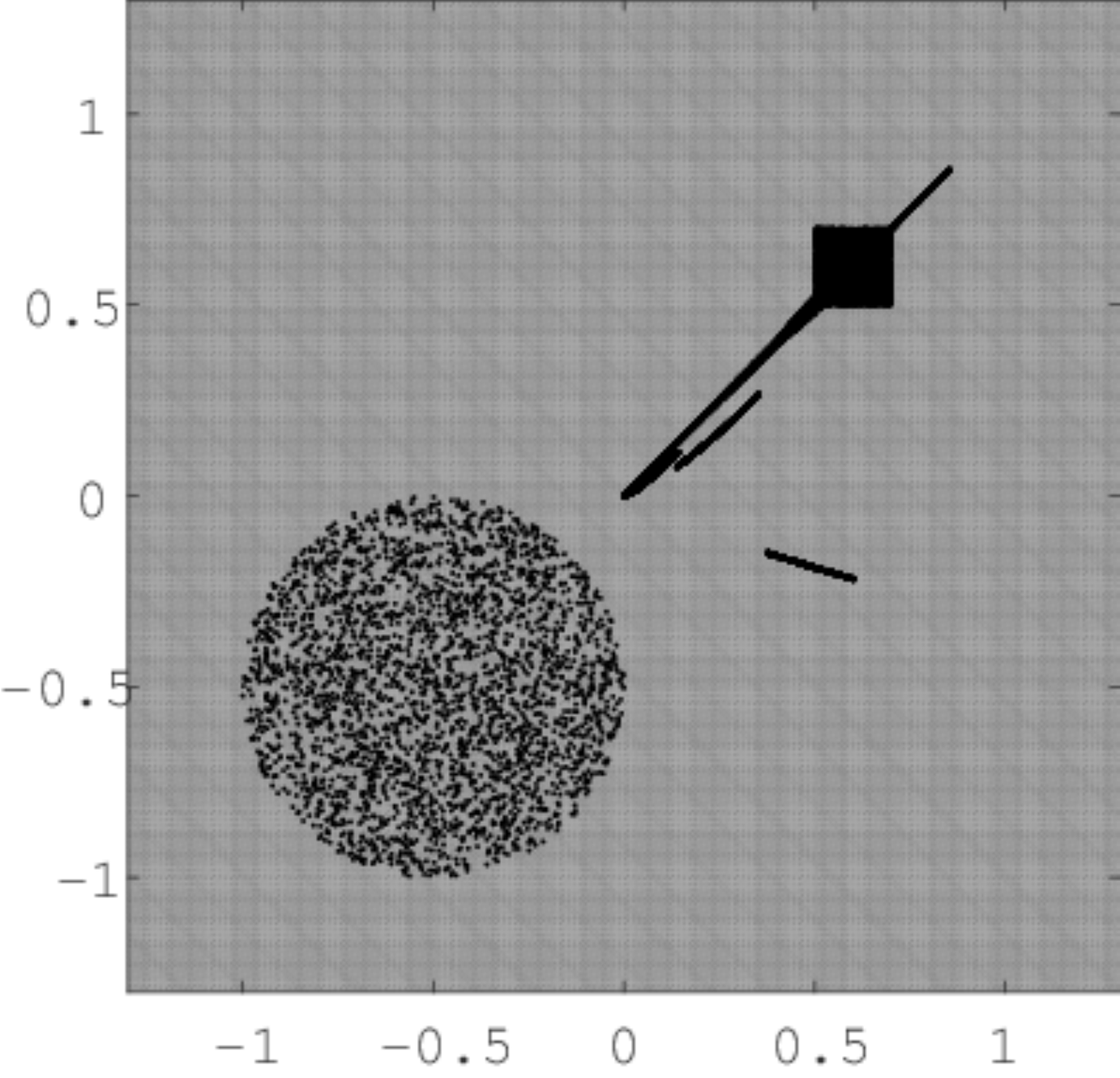}}
\subfigure[$m=4$]{
\includegraphics[scale=\sizesmallfig]{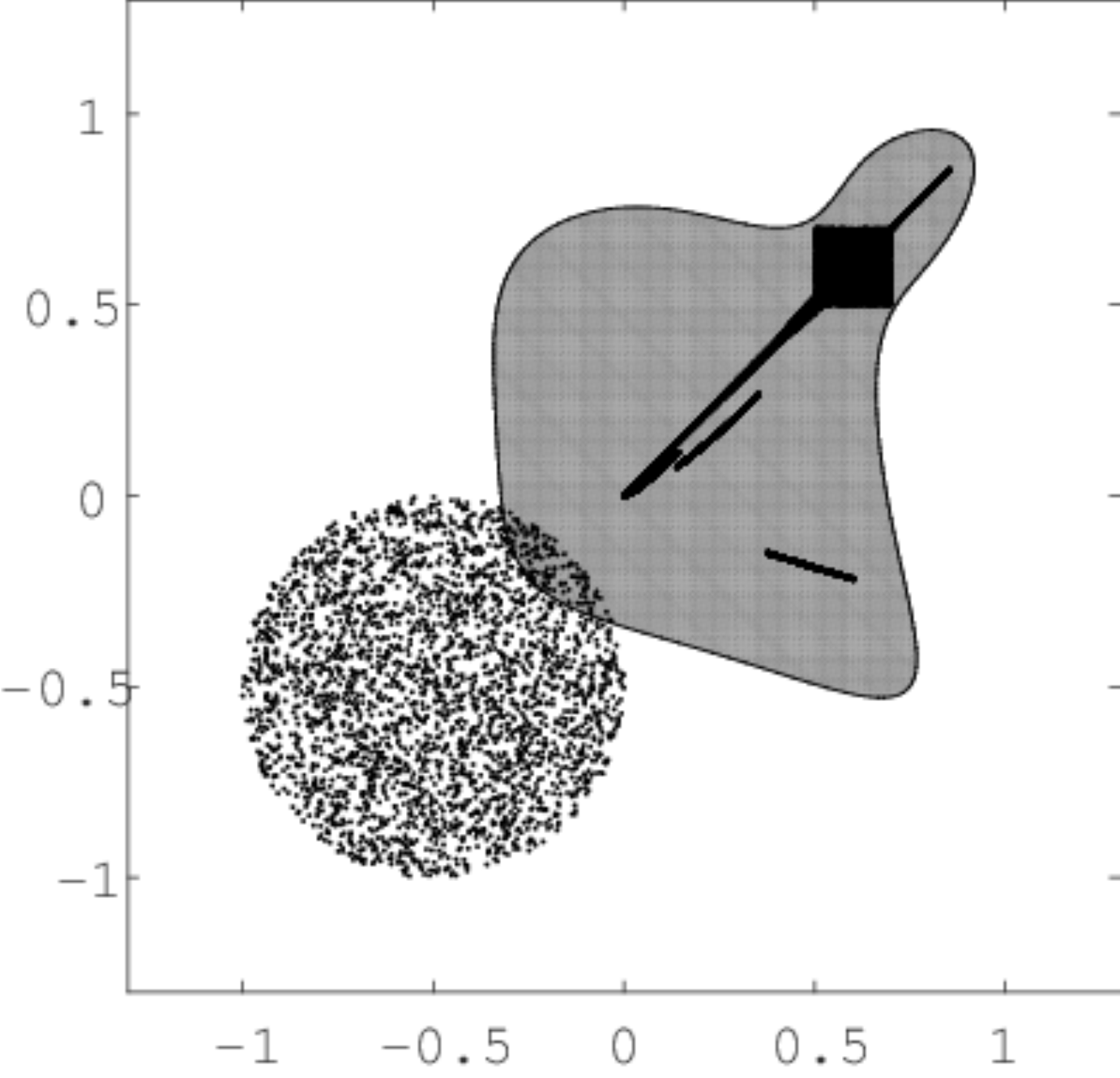}}
\subfigure[$m=5$]{
\includegraphics[scale=\sizesmallfig]{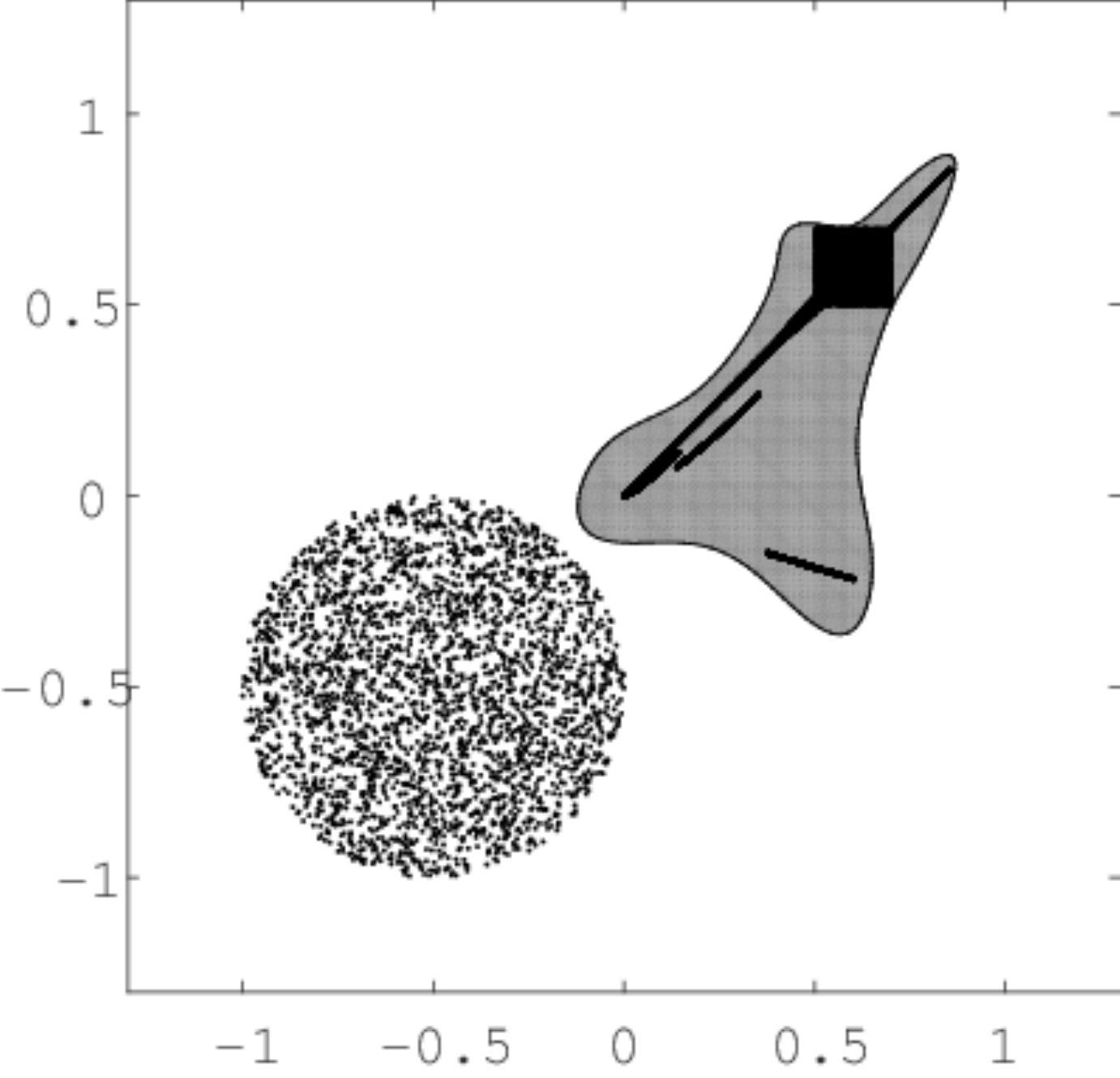}}
\caption{A hierarchy of sublevel sets $w_m^{\star}$ for Example~\ref{ex:testout}}	
\label{fig:testout}
\end{figure}

\section{Conclusion and Future Works}
\label{conclusion-future-works}
In this paper, we give a formal framework to relate the template generation 
problem to the property to prove on analyzed program : well-representative templates. 
We proposed a practical method to compute well-representative template bases in the case of polynomial arithmetic using sum-of-squares programming. This method is able to handle non trivial examples, as illustrated through the numerical experiments.

Topics of further investigation include refining the invariant bounds generated for 
a specific sublevel property, by applying the policy iteration algorithm. Such a refinement would be of particular interest if one can not decide whether the set of variables values avoids an unsafe region when the feasible invariant bound yields a negative value for $\alpha$. For the case of boundedness property, it would allow to decrease the value of the bounds on the variables. 
Finally, our method could be generalized to a larger class of programs, involving semialgebraic or transcendental assignments, by using the same polynomial reduction techniques as in~\cite{AGMW14nltemplates}.

\end{document}